\documentclass[12pt]{amsart}
\usepackage{epsfig}
\usepackage{amsfonts}
\usepackage{amsmath}
\newtheorem{thm}{Theorem}[section]

\newtheorem{prop}[thm]{Proposition}
\newtheorem*{prob*}{Problem}
\newtheorem*{thm*}{Theorem}

\theoremstyle{definition}

\newtheorem*{defn*}{Definition}
\newtheorem{rem}[thm]{Remark}

\newtheorem*{rem*}{Remark}
\numberwithin{equation}{section}

\newcommand{\C}{\mathbb C}
\newcommand{\R}{\mathbb R}

\DeclareMathOperator{\Res}{Res}
\DeclareMathOperator{\Ind}{Indep}
\DeclareMathOperator{\Laguerre}{Laguerre}
\DeclareMathOperator{\Bessel}{Bessel}
\DeclareMathOperator{\e}{e}

\DeclareMathOperator{\Mat}{Mat}

\newcommand{\re}{\mathop{\mathrm{Re}}}

\begin{document}
\title[Hard edge limit of two strongly coupled matrices]
 {\bf{Hard edge limit of the product of two strongly coupled random matrices}}

\author{Gernot Akemann}
\address{
Department of Physics,
Bielefeld University,  Postfach 100131, D-33501 Bielefeld, Germany
} \email{akemann@physik.uni-bielefeld.de}

\author{Eugene Strahov}
\address{Department of Mathematics, The Hebrew University of
Jerusalem, Givat Ram, Jerusalem
91904, Israel}\email{strahov@math.huji.ac.il}
\thanks{One of the authors (G.~A.) was partly supported
by LabEx PALM grant ANR-10-LABX- 0039-PALM and by the German research council DFG grant AK35/2-1. He
thanks the LPTMS Orsay and the Simons Center for Geometry and Physics, Stony Brook University, (program on Random Matrices 2015) for their kind hospitality where part of this work was done. Furthermore, fruitful discussions with Mario Kieburg are acknowledged.\\ }

\commby{}

\begin{abstract}
We investigate the hard edge scaling limit of the ensemble defined by the squared singular values
of the product of two coupled complex
random matrices.
When taking the coupling parameter to be dependent on the size of the product matrix, in a certain double scaling regime at the origin
the two matrices become strongly coupled and
we obtain a new hard edge limiting kernel. It
interpolates between the classical Bessel-kernel
describing the hard edge scaling limit of the Laguerre ensemble
of a single matrix,
and the Meijer G-kernel of Kuijlaars and Zhang
describing the hard edge scaling limit for the product of two independent Gaussian complex matrices.
It differs from the interpolating kernel of Borodin to which we compare as well.
\end{abstract}
\maketitle

\section{Introduction and results}
\subsection{Introduction}
\ \\

It is well known that the squared singular values of a complex rectangular Gaussian random matrix form a determinantal point process on $\R_{>0}$.
The correlation kernel of this determinantal point process  can be expressed in terms of the classical Laguerre polynomials, and can be studied
in different asymptotic regimes. In particular, the bulk, soft edge, and hard edge scaling limits of the correlation kernel
are
obtained
using known asymptotic expansions of the Laguerre polynomials, see, for example, Forrester \cite[Ch. 7]{ForresterLogGases}.
The hard edge scaling limit of the
kernel can be written in terms of Bessel functions, and it is called the Bessel-kernel.
The limiting determinantal point process defined by the Bessel-kernel was studied by many authors, see, for example,
Forrester \cite{Forrester93} and Tracy and Widom \cite{Tracy},
and we will use the representation of the former.
It enjoys important applications in Physics, including the Dirac operator spectrum in Quantum Chromodynamics \cite{SV93}.

Products of $M$ independent complex Gaussian random matrices have seen much progress recently, and we refer to \cite{AIp} and the thesis \cite{JespersThesis} for recent reviews including more general product matrices.
In particular one of the authors and his coworkers
\cite{AkemannIpsenKieburg,AkemannKieburgWei}
obtained a determinantal point process describing the squared singular values of such a product matrix.
It was shown by Kuijlaars and Zhang \cite{KuijlaarsZhang} that the kernel of
this point process
has  a remarkable hard edge scaling limit. The limiting
correlation kernel called Meijer G-kernel labelled by $M$
can be considered as a generalisation of the classical Bessel-kernel
with $M=1$.
It is universal as it appears in several other random matrix ensembles, including the Cauchy multi-matrix model
\cite{Bertola1,Bertola2,Bertola3}, in the Muttalib-Borodin ensemble
\cite{Muttalib,Borodin:1998} as was pointed out by Kuijlaars and Stivigny  \cite{ArnoDries}, and in the product of two coupled matrices as was shown by the authors \cite{AkemannStrahov}.

The question arises whether there is a family of determinantal point process that interpolates between the Bessel-kernel point process,
and the Meijer G-kernel point process for a given $M$. We note that
interpolating ensembles are of great interest in Random Matrix Theory.
A classical example is the ensemble solved by Pandey and Mehta \cite{PandeyMehta}
that interpolates between the Gaussian Orthogonal and Gaussian Unitary Ensemble
and thus describes the onset of time reversal symmetry breaking.
Another example is the
ensemble of Hermitian matrices introduced by  Moshe,  Neuberger, and Shapiro \cite{Moshe}, and studied
further by Johansson \cite{JohanssonGumbel}. This ensemble gives rise to a family of determinantal processes whose edge behaviour interpolates between the Poisson process of uncorrelated eigenvalues
and the Airy-kernel point process that appears in many applications.
For completeness one could also mention the works of Bohigas and coworkers \cite{Bohigas06,Bohigas09} on deformations of the Tracy-Widom distribution and their application to thinning processes. Also there interpolations between random matrix ensembles and classical statistical ensembles occur.

In this paper we present a new family of determinantal point processes that indeed interpolates between the Bessel-kernel
of a single Gaussian matrix $M=1$,
and the Meijer G-kernel
responsible for the hard edge scaling limit of the product of $M=2$ independent Gaussian complex matrices.
This new family arises from
the product of two coupled random matrices. The squared singular values of the product matrix of such matrices form a 
family of 
determinantal point process on $\R_{>0}$, as was shown in \cite{AkemannStrahov} for finite size $N$ of the product matrix.
Once we assume that the coupling parameter $\mu$ depends on $N$,
the asymptotic investigation of this point process
leads to several scaling regimes. In particular, if the coupling parameter
depends on $N$ as $\mu(N)=gN^{-1}$, then
the two matrices become strongly coupled.
The resulting hard edge scaling limit gives a new family of limiting determinantal processes parameterised by the parameter $g$, that is defined by the correlation kernel $\mathbb{S}(x,y;g)$. We show that
this family
is integrable in the sense of Its, Izergin, Korepin and Slavnov \cite{Its} and
can be considered as a deformation
of the Bessel-kernel point process, where
$g$ plays the role of the deformation parameter. This means that as $g\rightarrow 0$, the correlation kernel $\mathbb{S}(x,y;g)$ converges to the Bessel-kernel. Moreover, as $g\rightarrow\infty$ the correlation kernel $\mathbb{S}(x,y;g)$ (under a certain rescaling) converges to the corresponding Meijer G-kernel of the product of $M=2$ independent complex Gaussian matrices. Thus $\mathbb{S}(x,y;g)$
has the desired interpolation property.

The existence of such a limiting kernel
could have been expected for the following reason. When studying the complex eigenvalues instead of the singular values of two coupled matrices,
there exists the so-called weak non-Hermiticity limit \cite{FKS,FKS2}, where the two coupled matrices that are multiplied become almost complex conjugates of each other \cite{James}.
In the same sense as above, the limiting kernel of Osborn \cite{James} interpolates between real eigenvalues described by the Bessel-kernel, and the corresponding kernel of complex eigenvalues at strong non-Hermiticity (which can also be written in terms of Meijer G-functions).
In our case the singular values remain always real and positive for the entire family of kernels  $\mathbb{S}(x,y;g)$,
and the limit of strongly coupled matrices provides a new interpolating
process as well.

A further example should be mentioned here.
The family of kernels found by Borodin \cite{Borodin:1998}
originally given in terms of Wright's generalised Bessel functions also depends on a parameter $\theta\geq0$. It interpolates between the Bessel-kernel of the Laguerre ensemble
at $\theta=1$,
and the Meijer G-kernel
at $M=2$ with specific parameter values when choosing $\theta=1/2$ or
$\theta=2$ as it was shown by Kuijlaars and Stivigny \cite{ArnoDries}.
Initially the
Muttalib-Borodin ensemble \cite{Muttalib,Borodin:1998}
was introduced as an eigenvalue model, but
recently  upper triangular matrix representations have been constructed in \cite{Che} and \cite{FW}. Furthermore, in \cite{FW} an alternative double contour representation of Borodin's kernel was found for arbitrary real $\theta>0$. It extends
an alternative representation derived in \cite{LSZ} for 
$\theta=2$ 
where the Muttalib-Borodin ensemble reduces to a random matrix model of disordered bosons.
In the complex contour integral form
it is obvious that Borodin's kernel is
different from our limiting kernel given by a different double contour integral as well. Furthermore, they
lead to
different Meijer G-kernels at $M=2$.

Several open questions would be interesting to address. It is well known that the Fredholm determinant of the Bessel-kernel can be related to Painlev\'e V \cite{Tracy}. It is open if such
a relation can be established for the Meijer G-kernel,
see however
\cite{ES14} for progress in that directions. And
it is even less clear if
a parameter dependent relation to Painlev\'e equations exists for our interpolating kernel.
From the physical side it would be very interesting to see if our kernel appears in the singular value spectrum of Quantum Chromodynamics with iso-spin chemical potential which was studied in \cite{KW} on the level of an effective field theory.

The remainder of this article is organised as follows. In the following subsections we summarise the known results for the joint density (Subsection \ref{Ensembles}) and kernel (Subsection \ref{ExactFormulae}) of
the product of two coupled matrices \cite{AkemannStrahov}, two independent matrices \cite{AkemannIpsenKieburg,KuijlaarsZhang} and the well-known Laguerre ensemble of a single matrix, respectively. Our main results on the new limiting kernel (Subsection \ref{HardEdge})
which is integrable (Subsection \ref{integrable}) and interpolating between the latter two (Subsection \ref{Interpolating}) are given at the end of this section.
Before turning to the proofs, in Section \ref{Comparison} we illustrate our results by comparing with Borodin's kernel and by plotting the corresponding unfolded interpolating microscopic densities at different parameter values.
In Section \ref{Proof-Thm1.5} we present the proof for our new kernel,
in Section \ref{ProofInterpolatingKernel} the proof that it interpolates,
and in Section \ref{ProofInterpolatingDensity} the proofs for the corresponding density.
Appendix \ref{Heuristic} gives a heuristic argument for the limit from the interpolating kernel to the Bessel-kernel and Appendix \ref{AppB} provides further technical details.

\subsection{The ensembles}\label{Ensembles}
\subsubsection{Singular values for products of two coupled matrices}\label{SectionEnsemble1}
Let $A$, $B\in\Mat\left(\C,N\times M\right)$ be two independent matrices of size $N\times M$ with i.i.d.  standard normal complex Gaussian entries ${\mathcal N}(0,1/2)$. Set
\begin{equation}\label{X12def}
X_1=\frac{1}{\sqrt{2}}\left(A-i\sqrt{\mu}B\right),\;\; X_2=\frac{1}{\sqrt{2}}\left(A^{*}-i\sqrt{\mu}B^{*}\right),
\end{equation}
where $\mu\in(0,1)$ is a coupling parameter. We call $X_1$ and $X_2$ two coupled matrices (with the coupling parameter $\mu$). It was shown
by the authors in
\cite{AkemannStrahov} that the squared singular values of the product matrix $X_1X_2$ form a determinantal point process on $\R_{>0}$. Namely, the following
proposition holds true.
\begin{prop}\label{PropositionMainDensity}
Assume that $M\geq N$, and set
\begin{equation}\label{nu}
\nu=M-N.
\end{equation}
Then the joint probability density function for
the squared singular values $y_1$, $\ldots$, $y_N>0$ of the matrix $Y=X_1X_2$ is given by\footnote{Here and thoughout the article we use the notation $x^{\frac{1}{2}}=\sqrt{x}$ synonymously for positive real numbers $x>0$.}
\begin{equation}\label{TheMainJointProbabilityDensityFunction}
\begin{split}
&P(y_1,\ldots,y_N;\mu)=\frac{1}{Z_N}\det\left[y_i^{\frac{j-1}{2}}I_{j-1}\left(\frac{1-\mu}{\mu}{y_i}^{\frac{1}{2}}\right)\right]_{i,j=1}^N
\det\left[y_i^{\frac{j+\nu-1}{2}}K_{j+\nu-1}\left(\frac{1+\mu}{\mu}{y_i}^{\frac{1}{2}}\right)\right]_{i,j=1}^N,
\end{split}
\end{equation}
where
\begin{equation}\label{ZN}
Z_N=\frac{N!}{2^{N\nu+N^2}}\mu^N(1+\mu)^{N\nu}(1-\mu^2)^{\frac{N(N-1)}{2}}\prod\limits_{j=1}^N\Gamma(j)\Gamma(j+\nu).
\end{equation}
\end{prop}
In the proposition above $I_{\kappa}(z)$ denotes  the modified Bessel function of the first kind  defined by
\begin{equation}\label{BesselFunctionI}
I_{\kappa}(z)=
\sum\limits_{m=0}^{\infty}\frac{1}{m!\Gamma(\kappa+m+1)}
\left( \frac{z}{2}\right)^{2m+\kappa},
\end{equation}
valid for $z$ and $\kappa\in\C$. For $(z/2)^\kappa$ we choose the principal branch analytic on $\C\setminus(-\infty,0]$.
Above $K_{\kappa}(z)$ denotes the modified Bessel function of the second kind
and can be
defined by the integral
\begin{equation}
K_{\kappa}(z)=\frac{\Gamma\left(\kappa+\frac{1}{2}\right)(2z)^{\kappa}}{\sqrt{\pi}}\int\limits_0^{\infty}\frac{dt\,\cos (t)}{(t^2+z^2)^{\kappa^2+\frac{1}{2}}},
\ \ \re(\kappa)>-\frac12, \ \ |\arg(z)|<\frac{\pi}{2}\ .
\end{equation}
see, for example,
\cite{NIST}. Also here we choose the principal branch analytic on $\C\setminus(-\infty,0]$.

\subsubsection{Singular values for products of two independent complex Gaussian  matrices}\label{SectionComplexGaussianMatrices}
Let $X_1\in\Mat\left(\C,N\times M\right)$, and
$X_2\in\Mat\left(\C,M\times N\right)$
be two independent
matrices whose entries are i.i.d
standard normal complex Gaussian variables ${\mathcal N}(0,1/2)$. As before in Subsection \ref{SectionEnsemble1} assume that $M\geq N$, and
define $\nu$ by equation (\ref{nu}). It is known (see
\cite{AkemannIpsenKieburg}, formulae (18) and (21)) that the squared singular values $y_1,\ldots,y_N>0$ of the matrix $Y=X_1X_2$
have the joint density
\begin{equation}\label{TheTwoIndependentJointProbabilityDensityFunction}
P_{\Ind}(y_1,\ldots,y_N)=\frac{1}{Z_N^{\Ind}}\det\left[y_i^{j-1}\right]_{i,j=1}^N\det\left[2y^{\frac{j+\nu-1}{2}}K_{j+\nu-1}\left(2{y_i}^\frac{1}{2}\right)\right]_{i,j=1}^N,
\end{equation}
where $Z_N^{\Ind}=N!\prod_{j=1}^N\Gamma(j)^2\Gamma(j+\nu)$.  Equation (\ref{TheTwoIndependentJointProbabilityDensityFunction})
states that the squared singular values of the product matrix $Y$
form a determinantal point process.
Clearly, we have
\begin{equation}\label{FirstLimtingRelation}
\underset{\mu\rightarrow 1}{\lim}
P(y_1,\ldots,y_N;\mu)=P_{\Ind}(y_1,\ldots,y_N).
\end{equation}
This fact is obvious from the very definition of the ensembles.
It can also be seen directly from the explicit formula for the joint probability density function  $P(y_1,\ldots,y_N;\mu)$,
equation (\ref{TheMainJointProbabilityDensityFunction}), as shown
by the authors in \cite[Appendix A]{AkemannStrahov}.

Note that we are not considering the most general product of two independent matrices here. One could choose a different dimension for $X_2\in\Mat\left(\C,M\times N^\prime\right)$
with the joint distribution (\ref{TheTwoIndependentJointProbabilityDensityFunction}) then depending on two parameters $\nu_1=\nu=M-N$ and $\nu_2=N^\prime-N$, cf. \cite{AkemannIpsenKieburg}. Because we are interested in the interpolation between the ensembles of one and two random matrices we have to restrict the product matrix $Y=X_1X_2$ to be square, setting $\nu_2=0$ in \cite{AkemannIpsenKieburg}.

\subsubsection{The Laguerre ensemble}
The Laguerre ensemble relevant in the context of this paper is defined by the joint probability density function of the variables $y_1$, $\ldots$, $y_N>0$
\begin{equation}\label{TheLaguerreProbabilityDensityFunction}
P_{\Laguerre}(y_1,\ldots,y_N)
=\frac{2^{N(N+\nu-1)}}{N!\prod\limits_{j=1}^N\Gamma(j)\Gamma(j+\nu)}
\left(\det\left[y_i^{\frac{j-1}{2}}\right]_{i,j=1}^N\right)^2
\prod\limits_{i=1}^Ny_i^{\frac{\nu-1}{2}}\exp\left[-2y_i^{\frac{1}{2}}\right].
\end{equation}
Changing variables in equation (\ref{TheLaguerreProbabilityDensityFunction}),
\begin{equation}\label{change}
y_i\mapsto v_i=2y_i^{\frac{1}{2}},
\end{equation}
we map eq. (\ref{TheLaguerreProbabilityDensityFunction}) to the joint probability density function of the squared singular values $v_1,\ldots, v_N>0$ of a single random matrix $X_1\in\Mat\left(\C,N\times M\right)$ of size $N\times M$ with i.i.d. standard normal complex Gaussian entries ${\mathcal N}(0,1/2)$, the classical Laguerre ensemble. In these standard variables the joint density is given by
\begin{equation}\label{StandardLaguerrejpdf}
\frac{1}{N!\prod\limits_{j=1}^N\Gamma(j)\Gamma(j+\nu)}
\left(\det\left[v_i^{j-1}\right]_{i,j=1}^N\right)^2\prod\limits_{i=1}^Nv_i^\nu \e^{-v_i},
\end{equation}
see
\cite[Ch. 7]{ForresterLogGases}. To relate this ensemble to that defined in Subsection \ref{SectionEnsemble1} we use formula
(\ref{TheMainJointProbabilityDensityFunction}), and replace the modified Bessel functions inside the determinants by their large argument asymptotic expressions. A short calculation yields
\begin{equation}\label{SecondLimtingRelation}
\underset{\mu\rightarrow 0}{\lim}
P(y_1,\ldots,y_N;\mu)=P_{\Laguerre}(y_1,\ldots,y_N),
\end{equation}
see
\cite[Appendix A]{AkemannStrahov}
for details.

\subsection{Exact formulae for the correlation kernels}\label{ExactFormulae}
\subsubsection{The correlation kernel for the singular values of products of two coupled matrices}
Proposition \ref{PropositionMainDensity} implies that  the squared singular values $y_1$, $\ldots$, $y_N>0$ of
the product $Y=X_1X_2$ (where $X_1$ and $X_2$ are two coupled matrices) form a determinantal point process,
\begin{equation}\label{detPP}
P(y_1,\ldots,y_N;\mu)=
 \det\left[K_N(y_i,y_j;\mu)\right]_{i,j=1}^N\ .
\end{equation}
It was shown by the authors in
\cite{AkemannStrahov} that the correlation kernel of this process can be written as
\begin{equation}\label{K}
K_N(x,y;\mu)=\sum\limits_{n=0}^{N-1}P_n(x)Q_n(y),\ \ x,y>0,
\end{equation}
constituting a biorthogonal ensemble.
The functions $P_{n}(x)$ and  $Q_{n}(y)$ are defined
by
\begin{equation}\label{P1}
P_n(x)=(-1)^n\frac{(\nu+n)!n!}{\nu!}\left(\frac{1}{\mu}\right)^{\frac{1}{2}}
\sum\limits_{k=0}^n\left(\frac{2}{1-\mu}x^{\frac{1}{2}}\right)^k\frac{(-n)_k}{(\nu+1)_kk!}
I_k\left(\frac{1-\mu}{\mu}x^{\frac{1}{2}}\right),
\end{equation}
and
\begin{equation}\label{Q1}
Q_{n}(y)=(-1)^n\frac{2}{(n!)^2\nu!}\left(\frac{1}{\mu}\right)^{\frac{1}{2}}
\sum\limits_{l=0}^{n}\left(\frac{2}{1+\mu}y^{\frac{1}{2}}\right)^{l+\nu}
\frac{(-n)_l}{(\nu+1)_ll!}K_{l+\nu}\left(\frac{1+\mu}{\mu}y^{\frac{1}{2}}\right).
\end{equation}
Here $(a)_n=\Gamma(a+n)/\Gamma(a)=a(a+1)\cdots(a+n-1)$ denotes the Pochhammer symbol for $n\in\mathbb{N}$. In particular for negative integers $-a=k\in\mathbb{N}$ with
$k\geq n$ we have $(-k)_n=(-1)^nk!/(k-n)!$.
The starting point for the subsequent asymptotic analysis is
a Christoffel-Darboux
type formula for the correlation kernel $K_N(x,y;\mu)$.
\begin{prop}\label{TheoremChristoffelDarboux}
The Christoffel-Darboux type formula for the correlation kernel $K_{N}(x,y;\mu)$ is
for $N\geq 2$ and $x\neq y$ given by
\begin{equation}\label{CDKernel}
\begin{split}
K_{N}(x,y;\mu)=-\frac{a_{-2,N}P_{N-2}(x)Q_{N}(y)+a_{-2,N+1}P_{N-1}(x)Q_{N+1}(y)+a_{-1,N}P_{N-1}(x)Q_{N}(y)}{x-y}\\
\ \ +\frac{a_{1,N-1}P_{N}(x)Q_{N-1}(y)+a_{2,N-2}P_{N}(x)Q_{N-2}(y)+a_{2,N-1}P_{N+1}(x)Q_{N-1}(y)}{x-y},
\end{split}
\end{equation}
where the coefficients $a_{-2,N}$, $a_{-1,N}$, $a_{1,N}$ and $a_{2,N}$
read
\begin{equation}\label{xPn6}
\begin{split}
a_{2,N}&=\frac{(1-\mu)^2}{4(N+2)(N+1)},\\
a_{1,N}&=\mu+\frac{(1-\mu)^2(2N+\nu+2)}{2(N+1)},\\
a_{-1,N}&=\mu N^2(N+\nu)(3N+\nu)+\frac{(1-\mu)^2}{2}N^2(\nu+2N)(\nu+N),\\
a_{-2,N}&=\mu N^2(N-1)^2(N+\nu)(N+\nu-1)
+\frac{(1-\mu)^2}{4}(\nu+N)(\nu+N-1)N^2(N-1)^2.
\end{split}
\end{equation}
\end{prop}
\begin{proof}
See \cite[Thm. 3.6]{AkemannStrahov}.
\end{proof}
Next we will need the following contour integral representations for $P_n(x)$ and $Q_n(y)$, and for the correlation kernel $K_{N}(x,y;\mu)$.
\begin{prop}\label{PropositionPQContourIntegralRepresentation}
We have
for $x,y>0$
\begin{equation}
\label{Pn}
P_n(x)=\frac{\Gamma(\nu+n+1)\Gamma^2(n+1)}{\mu^{\frac{1}{2}}2\pi i}\oint\limits_{\Sigma_n}dt
\frac{\Gamma(t-n)\left(\frac{2}{1-\mu}x^{\frac{1}{2}}\right)^t}{\Gamma(\nu+1+t)\Gamma(t+1)}I_t\left(\frac{1-\mu}{\mu}x^{\frac{1}{2}}\right),
\end{equation}
and
\begin{equation}
\label{Qn}
Q_n(y)=\frac{2}{\mu^{\frac{1}{2}}\Gamma(n+1)2\pi i}\oint\limits_{\Sigma_n}ds
\frac{\Gamma(s-n)\left(\frac{2}{1+\mu}y^{\frac{1}{2}}\right)^{s+\nu}}{\Gamma(\nu+1+s)\Gamma(s+1)}K_{s+\nu}\left(\frac{1+\mu}{\mu}y^{\frac{1}{2}}\right).
\end{equation}
In the formulae just written above $\Sigma_n$ is a closed contour that
encircles $0,1,\ldots, n$ once in
counterclockwise
direction,
such that $\re(s+\nu)>-1/2$.
\end{prop}
\begin{proof}
Use the Residue Theorem, and the fact that $\underset{z=k}{\Res}\;\Gamma(z-n)=\frac{(-1)^{n-k}}{(n-k)!}$.
The representation (\ref{Pn}) can also be inferred from
\cite[Prop. 3.7]{AkemannStrahov},
using
$I_\nu(z)=\frac{1}{\Gamma(\nu+1)}\left(\frac{z}{2}\right)^\nu {}_0F_1\left(-;\nu+1;\left(\frac{z}{2}\right)^2\right)$.
\end{proof}
Eq. (\ref{Qn}) is an alternative  to the contour integral representation in  \cite[Prop. 3.7]{AkemannStrahov} containing  Gauss' hypergeometric function.

\begin{prop}\label{PropositionContourRepresentationCorrelationKernel} The correlation kernel $K_N(x,y;\mu)$ admits the following representation
\begin{equation}\label{MainExactKernel}
\begin{split}
K_N(x,y;\mu)=&\frac{1}{(2\pi i)^2\mu(x-y)}
\oint\limits_{\Sigma_N}dt\oint\limits_{\Sigma_N}ds
\frac{\Gamma(-t)\left(\frac{2}{1-\mu}x^{\frac{1}{2}}\right)^t\Gamma(-s)\left(\frac{2}{1+\mu}y^{\frac{1}{2}}\right)^{s+\nu}}{\Gamma(t+\nu+1)\Gamma(s+\nu+1)}
\\
&\times I_t\left(\frac{1-\mu}{\mu}x^{\frac{1}{2}}\right)K_{s+\nu}\left(\frac{1+\mu}{\mu}y^{\frac{1}{2}}\right)
A_N(s,t;\mu)\frac{\Gamma(\nu+N+1)\Gamma(N+1)}{\Gamma(1-t+N)\Gamma(1-s+N)},
\end{split}
\end{equation}
where $\Sigma_N$ is a closed contour that encircles $0,1,\ldots, N$ once in
counterclockwise
direction
with $\re(s+\nu)>-1/2$, and
\begin{equation}\label{MainPolynomial}
\begin{split}
A_N(s,t;\mu)=&-\frac{(1+\mu)^2}{4}(t-N)(t-N+1)
-\frac{(1+\mu)^2}{4}(\nu+N+1)(N+1)\frac{(t-N)}{(s-N-1)}\\
&-\mu(3N+\nu)(t-N)-\frac{(1-\mu)^2}{2}(2N+\nu)(t-N)\\
&+\mu N(s-N)+\frac{(1-\mu)^2}{2}(2N+\nu)(s-N)\\
&+\frac{(1-\mu)^2}{4}(s-N)(s-N+1)
+\frac{(1-\mu)^2}{4}(N+1)(\nu+N+1)\frac{(s-N)}{(t-N-1)}.
\end{split}
\end{equation}
In particular we have at equal arguments
\begin{equation}\label{MainExactKernelyy}
\begin{split}
K_N(y,y;\mu)=&\frac{1}{(2\pi i)^2\mu^2}
\oint\limits_{\Sigma_N}dt\oint\limits_{\Sigma_N}ds
\frac{\Gamma(-t)\left(\frac{2}{1-\mu}y^{\frac{1}{2}}\right)^{t-1}\Gamma(-s)\left(\frac{2}{1+\mu}y^{\frac{1}{2}}\right)^{s+\nu}}{\Gamma(t+\nu+1)\Gamma(s+\nu+1)}
\\
&\times I_{t-1}\left(\frac{1-\mu}{\mu}y^{\frac{1}{2}}\right)K_{s+\nu}\left(\frac{1+\mu}{\mu}y^{\frac{1}{2}}\right)
A_N(s,t;\mu)\frac{\Gamma(\nu+N+1)\Gamma(N+1)}{\Gamma(1-t+N)\Gamma(1-s+N)}.
\end{split}
\end{equation}
\end{prop}
\begin{proof}
Use the Christoffel-Darboux type formula for the correlation kernel $K_N(x,y;\mu)$ (see Proposition \ref{TheoremChristoffelDarboux}), and
the contour integral representation for the functions $P_n(x)$ and $Q_n(y)$ (see Proposition \ref{PropositionPQContourIntegralRepresentation}) to obtain eq. (\ref{MainExactKernel}). Clearly due to eq. (\ref{K}) the kernel is regular at equal arguments $x=y$. Hence the integral in eq.  (\ref{MainExactKernel}) that follows from eq. (\ref{CDKernel}) at equal arguments must vanish at equal arguments to compensate the pole at $x=y$ in front of the integral,
\begin{equation}\label{Integrand}
\begin{split}
0=&\frac{1}{(2\pi i)^2\mu}
\oint\limits_{\Sigma_N}dt\oint\limits_{\Sigma_N}ds
\frac{\Gamma(-t)\left(\frac{2}{1-\mu}y^{\frac{1}{2}}\right)^t\Gamma(-s)\left(\frac{2}{1+\mu}y^{\frac{1}{2}}\right)^{s+\nu}}{\Gamma(t+\nu+1)\Gamma(s+\nu+1)}
\\
&\times I_t\left(\frac{1-\mu}{\mu}y^{\frac{1}{2}}\right)K_{s+\nu}\left(\frac{1+\mu}{\mu}y^{\frac{1}{2}}\right)
A_N(s,t;\mu)\frac{\Gamma(\nu+N+1)\Gamma(N+1)}{\Gamma(1-t+N)\Gamma(1-s+N)}.
\end{split}
\end{equation}
A simple Taylor expansion of the function $\left(\frac{2}{1-\mu}x^{\frac{1}{2}}\right)^tI_t\left(\frac{1-\mu}{\mu}x^{\frac{1}{2}}\right)$ at $x=y$ inside the integral (\ref{MainExactKernel}) in the limit $x\to y$ leads to the desired eq.  (\ref{MainExactKernelyy}). Here we have simply taken the limit under the integral. This can be justified by Taylor expanding $P_n(x)$ in eq. (\ref{CDKernel}) around $x=y$, using the finite sum in eq. (\ref{P1}) for the expansion, and then applying a contour integral representation for $P^\prime_n(x)$ in analogy to eq. (\ref{Pn}).
\end{proof}
\begin{rem}
In \cite{AkemannStrahov} two alternative double contour representations for the kernel (\ref{MainExactKernel}) were derived. They are based on a different contour integral representation of the function $Q_n(y)$, inserted into the Christoffel-Darboux type formula (\ref{CDKernel}).
A further resummation leads to a nested sum of double contour integrals in \cite[Thm. 3.8]{AkemannStrahov}, which is however not used in the asymptotic analysis.
\end{rem}
\subsubsection{The correlation kernel for the singular values of products of two independent matrices}
Denote by $K_N^{\Ind}(x,y)$ the correlation kernel of the determinantal point process on $\R_{>0}$ formed by the squared singular values of $Y=X_1X_2$,
where $X_1$ and $X_2$ are two independent complex Gaussian matrices defined in Subsection \ref{SectionComplexGaussianMatrices}. The correlation kernel
$K_N^{\Ind}(x,y)$ can be represented in different ways.
In particular,  it was shown
\cite[Prop. 5.1]{KuijlaarsZhang},
including more general products of independent random matrices,
that $K_N^{\Ind}(x,y)$ admits a double contour integral representation.
Namely,
\begin{equation}\label{KNInd}
K_N^{\Ind}(x,y)=\frac{1}{(2\pi i)^2}\int\limits_{-\frac{1}{2}-i\infty}^{-\frac{1}{2}+i\infty}ds\oint\limits_{\Sigma_N}dt\frac{\Gamma^2(s+1)}{\Gamma^2(t+1)}\frac{\Gamma(s+\nu+1)}{\Gamma(t+\nu+1)}
\frac{\Gamma(t-N+1)}{\Gamma(s-N+1)}\frac{x^ty^{-s-1}}{(s-t)},
\end{equation}
where the integration contour $\Sigma_N$ is defined in the same way as in Proposition \ref{PropositionContourRepresentationCorrelationKernel}.

Using equation (\ref{KNInd}), it is not
difficult to find the hard edge scaling limit of $K_N^{\Ind}(x,y)$. Namely, the following limiting relation holds true
uniformly for $x$, $y$ in compact subsets of the positive real axis
(see
\cite[Thm. 5.3]{KuijlaarsZhang})
\begin{equation}
\underset{N\rightarrow\infty}\lim\left(\frac{1}{N}K_N^{\Ind}\left(\frac{x}{N},\frac{y}{N}\right)\right)=S_{\nu,0}^{\Ind}(x,y),
\end{equation}
where the limiting kernel $S_{\nu_1,\nu_2}^{\Ind}(x,y)$ for the most general product of two independent rectangular Gaussian matrices is defined by the formula
\begin{equation}\label{SIndependent}
S_{\nu_1,\nu_2}^{\Ind}(x,y)=\frac{1}{(2\pi i)^2}\int\limits_{-\frac{1}{2}-i\infty}^{-\frac{1}{2}+i\infty}ds\int\limits_{\Sigma}dt\frac{\Gamma(s+1)\Gamma(s+\nu_1+1)\Gamma(s+\nu_2+1)}{\Gamma(t+1)\Gamma(t+\nu_1+1)\Gamma(t+\nu_2+1)}
\frac{\sin(\pi s)}{\sin(\pi t)}\frac{x^ty^{-s-1}}{(s-t)}.
\end{equation}
The contour $\Sigma$ starts at $+\infty$ in the upper half plane, encircles the positive real axis keeping $\re(t)>-\frac12$ to avoid the second contour, and returns to  $+\infty$ in the lower
half plane. We will need the form with general index pair $\nu_1,\nu_2$ later.
As it is shown by Kuijlaars and Zhang \cite[Thm. 5.3]{KuijlaarsZhang}), this limiting kernel can be also written as
\begin{equation}\label{SIndependentDefiniteIntegral}
S_{\nu_1,\nu_2}^{\Ind}(x,y)=\int\limits_0^1 du\ G^{1,0}_{0,3}\left(\begin{array}{ccc}
                       & - &  \\
                      0, & -\nu_1, & -\nu_2
                    \end{array}
\biggr|ux\right)G^{2,0}_{0,3}\left(\begin{array}{ccc}
                       & - &  \\
                      \nu_1,&\nu_2, & 0
                    \end{array}
\biggr|uy\right).
\end{equation}
Here
$G^{n,0}_{0,q}\left(\begin{array}{ccc}
                       & - &  \\
                      a, & b, & c
                    \end{array}
\biggr| z\right)$
is a Meijer G-function  with corresponding parameters, see e.g.
\cite{GradshteinRyzhik} for its definition.
\subsubsection{The correlation kernel for the Laguerre ensemble}
It is well known (see, for example,
\cite{ForresterLogGases}) that the correlation kernel for the Laguerre ensemble can be written in terms of the Laguerre polynomials.
Namely, denote by $K_N^{\Laguerre}(x,y)$ the correlation kernel for the ensemble defined by equation (\ref{TheLaguerreProbabilityDensityFunction}).
Using standard methods of Random Matrix Theory we find after the change of variables (\ref{change})
\begin{equation}\label{KLaguerre}
\begin{split}
K_N^{\Laguerre}(x,y)=&\sum\limits_{n=0}^{N-1}\frac{\Gamma(n+1)}{\Gamma(n+1+\nu)}L_n^{(\nu)}\left(2x^{\frac{1}{2}}\right)L_n^{(\nu)}\left(2y^{\frac{1}{2}}\right)
\left[\left(2x^{\frac{1}{2}}\right)^{\frac{\nu}{2}}\frac{\e^{-x^{\frac{1}{2}}}}{x^{\frac{1}{4}}}\right]
\left[\left(2y^{\frac{1}{2}}\right)^{\frac{\nu}{2}}\frac{\e^{-y^{\frac{1}{2}}}}{y^{\frac{1}{4}}}\right]\\
&\times\left(\frac{y^{\frac{1}{2}}}{x^{\frac{1}{2}}}\right)^{\frac{\nu}{2}}\frac{\e^{\frac{x^{\frac{1}{2}}}{\mu}}}{\e^{\frac{y^{\frac{1}{2}}}{\mu}}}.
\end{split}
\end{equation}
Here the factors in the second line take the form $h(y)/h(x)$. They have been added to the standard kernel of the Laguerre ensemble, without any effect for the following reason.
The density correlation functions of singular values are given by the determinant of this kernel and are thus the same as for the kernel without these factors, as they cancel out.
The two determinantal point processes are thus equivalent.

Applying the Christoffel-Darboux identity for the Laguerre polynomials we can rewrite the sum in the formula (\ref{KLaguerre}) for
$K_N^{\Laguerre}(x,y)$ in  a closed form. This gives
\begin{equation}\label{Kmu0}
\begin{split}
K_N^{\Laguerre}(x,y)=&-\frac{\Gamma(N+1)}{\Gamma(N+\nu)}\frac{L_N^{(\nu)}\left(2x^{\frac{1}{2}}\right)L_{N-1}^{(\nu)}\left(2y^{\frac{1}{2}}\right)
-L_N^{(\nu)}\left(2y^{\frac{1}{2}}\right)L_{N-1}^{(\nu)}\left(2x^{\frac{1}{2}}\right)}{2x^{\frac{1}{2}}-2y^{\frac{1}{2}}}\\
&\times\left[\left(2x^{\frac{1}{2}}\right)^{\frac{\nu}{2}}\frac{\e^{-x^{\frac{1}{2}}}}{x^{\frac{1}{4}}}\right]
\left[\left(2y^{\frac{1}{2}}\right)^{\frac{\nu}{2}}\frac{\e^{-y^{\frac{1}{2}}}}{y^{\frac{1}{4}}}\right]
\left(\frac{y^{\frac{1}{2}}}{x^{\frac{1}{2}}}\right)^{\frac{\nu}{2}}\frac{\e^{\frac{x^{\frac{1}{2}}}{\mu}}}{\e^{\frac{y^{\frac{1}{2}}}{\mu}}}.
\end{split}
\end{equation}
The hard edge scaling limit of the kernel $K_N^{\Laguerre}(x,y)$ leading to the Bessel-kernel can be obtained by standard methods, using
well-known
asymptotic formulae for the
uniform convergence of the rescaled
Laguerre polynomials
near the orign, see
\cite[Sec. 7.2.1]{ForresterLogGases}. We have
\begin{equation}\label{BesselLim}
\underset{N\rightarrow\infty}\lim\left(\frac{1}{N^2}K_N^{\Laguerre}\left(\frac{x^2}{4N^2},\frac{y^2}{4N^2}\right)\right)=S^{\sqrt{\Bessel}}(x,y),
\end{equation}
where
\begin{equation}\label{BesselKernel}
S^{\sqrt{\Bessel}}(x,y)= -
\frac{2x^{\frac{1}{2}}J_{\nu-1}\left(2x^{\frac{1}{2}}\right)J_{\nu}\left(2y^{\frac{1}{2}}\right)
-2y^{\frac{1}{2}}J_{\nu-1}\left(2y^{\frac{1}{2}}\right)
J_{\nu}\left(2x^{\frac{1}{2}}\right)}{(x-y)x^{\frac{1}{2}}y^{\frac{1}{2}}}\left(\frac{y}{x}\right)^{\frac{\nu}{2}}.
\end{equation}
We have introduced a superscript $({\scriptsize \sqrt{\Bessel}})$ to indicate that this representation
is obtained after the change of variables (\ref{change}).
It differs from the standard Bessel-kernel \cite{ForresterLogGases} by a factor of $4/(xy)^\frac12$. This is due to the square root of the Jacobian for the limit (\ref{BesselLim}) in terms of squared variables.
Due to identities among Bessel functions (see eq. (\ref{BesselId})) an alternative and equivalent form of the Bessel-kernel exists, replacing $J_{\nu-1}$ by $-J_{\nu+1}$. We emphasise that this is the limiting kernel for the singular values of a single random
matrix.
\subsection{The hard edge scaling limit of $K_N(x,y;\mu)$}\label{HardEdge}
\ \\

In what follows we drop the assumption that $\mu$ is constant, and consider the situation when
this parameter is a function of $N$ taking values in the interval $\mu\in(0,1)$. For simplicity, set
$$
\mu(N)=gN^{-\kappa},
$$
where
$g$ is some fixed positive constant, and $\kappa\geq0$.
Depending on the values of $\kappa$ we expect to obtain different hard edge scaling limits of our correlation kernel
$K_N(x,y;\mu(N))$.

Now we  formulate the  main result of the present paper.
\begin{thm}\label{TheoremTransitionToBesselKernel}
Assume that $\mu(N)=gN^{-\kappa}$, where $g>0$ is a positive constant. Assume that $x$,$y$ take values in a compact subset of $(0,+\infty)$.
Then we have\\
\textbf{(a)} For
$\kappa>1$
\begin{equation}
\begin{split}
&\underset{N\rightarrow\infty}{\lim}
\left[\frac{1}{N^2}K_N\left(\frac{x^2}{4N^2},\frac{y^2}{4N^2};
\mu=\frac{g}{N^{\kappa}}
\right) \e^{-\frac{x}{2N\mu(N)}+\frac{y}{2N\mu(N)}}\right]
=S^{\sqrt{\Bessel}}(x,y).
\end{split}
\nonumber
\end{equation}

\noindent
\textbf{(b)} For $\kappa=1$
\begin{equation}
\begin{split}
&\underset{N\rightarrow\infty}{\lim}
\left[\frac{1}{N^2}K_N\left(\frac{x^2}{4N^2},\frac{y^2}{4N^2};\mu=\frac{g}{N}
\right)\right]=\mathbb{S}(x,y;g),
\end{split}
\nonumber
\end{equation}
where the kernel $\mathbb{S}(x,y;g)$ is defined by
\begin{equation}\label{Skernel}
\begin{split}
\mathbb{S}(x,y;g)=&\frac{4}{(x^2-y^2)g}
\frac{1}{(2\pi i)^2}\oint\limits_{\Sigma}dt\oint\limits_{\Sigma}ds
\frac{\Gamma(-t)\Gamma(-s)x^ty^{s+\nu}}{\Gamma(t+\nu+1)\Gamma(s+\nu+1)}\\
&\times\left[\mathcal{P}(s,t,\nu)-g(s^2+t^2+\nu s-st)\right]I_t\left(\frac{x}{2g}\right)K_{s+\nu}\left(\frac{y}{2g}\right),
\end{split}
\end{equation}
and where we have introduced the following polynomial
\begin{equation}\label{MainPolynomial2}
\mathcal{P}(s,t,\nu)=\frac{1}{4}(t-s)\biggl(
t^2+s^2+(t+s)(\nu-1)-\nu\biggr).
\end{equation}
The contour $\Sigma$ starts at $+\infty$ in the upper half plane, encircles
the positive real axis, and returns to $+\infty$ in the lower half plane.
We call $\mathbb{S}(x,y;g)$ the kernel in the strongly correlated limit.\\

\noindent
\textbf{(c)}
For $1>\kappa\geq0$
\begin{equation}
\begin{split}
&\underset{N\rightarrow\infty}{\lim}
\left[\frac{\mu}{N}K_N\left(\frac{\mu x}{N},\frac{\mu y}{N};\mu=gN^{-\kappa}
\right)\right]=S_{\nu,0}^{\Ind}(x,y).
\end{split}
\nonumber
\end{equation}
\end{thm}
\begin{rem}
1) In what follows we will show that Theorem \ref{TheoremTransitionToBesselKernel} (a) holds true in two different ways. First, we will use a contour integral representation for the correlation
kernel suitable for the asymptotic analysis as $N\rightarrow\infty$. This method is especially important for us because it opens
the possibility to investigate the different asymptotic regimes including Theorem \ref{TheoremTransitionToBesselKernel} (b).
The second method uses heuristic arguments and is collected in Appendix \ref{Heuristic}.
There we will initially assume that $N$ is fixed. Then we will observe that as $\mu$ goes to $0$
our kernel $K_N(x,y;\mu)$  turns to that of a Laguerre-type ensemble, whose large $N$ asymptotics is well known. The fact that we get the same answer shows that the limits $\mu\to0$ and $N\to\infty$ commute in this regime.\\
2) Note that the above polynomial in eq. (\ref{MainPolynomial2}) has an equivalent form to be used later, as can be easily seen:
\begin{equation}\label{MainPolynomial3}
\begin{split}
\mathcal{P}(s,t,\nu)=&-\frac{1}{4}\biggl(s(s+\nu)(s+\nu-1)-t(t+\nu)(t+\nu-1)\\
&-\nu s(s+\nu)+\nu t(t+\nu)+t(t+\nu)s-s(s+\nu)t\biggr).
\end{split}
\end{equation}
3) Theorem \ref{TheoremTransitionToBesselKernel} (c) was proved in our previous work
for $\kappa=0$ only, see \cite[Thm. 3.9]{AkemannStrahov}.\\
\end{rem}
\subsection{Integrable form of the kernel in the strongly correlated limit}
\label{integrable}
\ \\

Recall that a correlation kernel $K(x,y)$ is called integrable (in the sense of Its, Izergin, Korepin and Slavnov \cite{Its}) if it can be represented
as
\begin{equation}\label{IntegrableForm}
K(x,y)=\frac{\sum\limits_{i=1}^kF_i(x)G_i(y)}{x-y},\;\;\;\mbox{where}\;\;\sum\limits_{i=1}^kF_i(x)G_i(x)=0.
\end{equation}
Integrable kernels lead to the theory of integrable Fredholm operators, which has many applications in different areas of mathematics and mathematical physics,
see
\cite{Its1},
\cite{ItsHarnad},
\cite{DeiftIntegrableOperators},
\cite{DIZ}
and
references therein. We argue that the new limiting kernel $\mathbb{S}(x,y;g)$ can be represented in an integrable form (\ref{IntegrableForm}),
with $k=4$.
Indeed, in formula (\ref{Skernel})
when using the representation (\ref{MainPolynomial3})
the integral over $s$ can be rewritten in terms of four functions
$F_1(y;g)$, $F_2(y;g)$, $F_3(y;g)$ and $F_4(y;g)$ defined
for $y>0$ by
\begin{equation}\label{FunctionF1}
F_1(y;g)=\frac{1}{2\pi i}\oint\limits_{\Sigma}ds\frac{\Gamma(-s)}{\Gamma(s+\nu+1)}y^{s+\nu}K_{s+\nu}\left(\frac{y}{2g}\right),
\end{equation}
\begin{equation}
F_2(y;g)=\frac{1}{2\pi i}\oint\limits_{\Sigma}ds\frac{s\Gamma(-s)}{\Gamma(s+\nu+1)}y^{s+\nu}K_{s+\nu}\left(\frac{y}{2g}\right),
\end{equation}
\begin{equation}
F_3(y;g)=\frac{1}{2\pi i}\oint\limits_{\Sigma}ds\frac{s(s+\nu)\Gamma(-s)}{\Gamma(s+\nu+1)}y^{s+\nu}K_{s+\nu}\left(\frac{y}{2g}\right),
\end{equation}
and
\begin{equation}\label{FunctionF4}
F_4(y;g)=\frac{1}{2\pi i}\oint\limits_{\Sigma}ds\frac{s(s+\nu)(s+\nu-1)\Gamma(-s)}{\Gamma(s+\nu+1)}y^{s+\nu}K_{s+\nu}\left(\frac{y}{2g}\right).
\end{equation}
On the other hand, the integral over $t$ in formula (\ref{Skernel}) 
together with eq. (\ref{MainPolynomial3}) 
can be written in terms of four functions $\Phi_1(x;g)$, $\Phi_2(x;g)$,
$\Phi_3(x;g)$, and $\Phi_4(x;g)$ defined
for $x>0$
by
\begin{equation}\label{Phi1}
\Phi_1(x;g)=\frac{1}{2\pi i}\oint\limits_{\Sigma}dt\frac{\Gamma(-t)}{\Gamma(t+\nu+1)}x^tI_t\left(\frac{x}{2g}\right),
\end{equation}
\begin{equation}
\Phi_2(x;g)=\frac{1}{2\pi i}\oint\limits_{\Sigma}dt\frac{t\Gamma(-t)}{\Gamma(t+\nu+1)}x^tI_t\left(\frac{x}{2g}\right),
\end{equation}
\begin{equation}
\Phi_3(x;g)=\frac{1}{2\pi i}\oint\limits_{\Sigma}dt\frac{t(t+\nu)\Gamma(-t)}{\Gamma(t+\nu+1)}x^tI_t\left(\frac{x}{2g}\right),
\end{equation}
and
\begin{equation}\label{Phi4}
\Phi_4(x;g)=\frac{1}{2\pi i}\oint\limits_{\Sigma}dt\frac{t(t+\nu)(t+\nu-1)\Gamma(-t)}{\Gamma(t+\nu+1)}x^tI_t\left(\frac{x}{2g}\right).
\end{equation}
With this notation,  formula (\ref{Skernel}) for the correlation kernel  $\mathbb{S}(x,y;g)$ can be written as
\begin{equation}\label{SIntegrableForm}
\begin{split}
\mathbb{S}(x,y;g)=&-\frac{1}{g(x^2-y^2)}\biggl[\Phi_1(x;g)F_4(y;g)-\Phi_4(x;g)F_1(y;g)\\
&-\nu\Phi_1(x;g)F_3(y;g)+\nu\Phi_3(x;g)F_1(y;g) + \Phi_3(x;g)F_2(y;g)-\Phi_2(x;g)F_3(y;g)\biggr]\\
&-\frac{4}{x^2-y^2}\biggl[\Phi_1(x;g)F_3(y;g)+\Phi_3(x;g)F_1(y;g)\\
&-\nu\Phi_2(x;g)F_1(y;g)-\Phi_2(x;g)F_2(y;g)\biggr].
\end{split}
\end{equation}
After
a change of variables, the new kernel  $\mathbb{S}(x,y;g)$ takes  an integrable form (\ref{IntegrableForm}),
with $k=4$. It is not obvious to directly link eq. (\ref{SIntegrableForm}) to the limit of the Christoffel-Darboux type formular (\ref{CDKernel}).

\subsection{The interpolating process}\label{Interpolating}
\ \\

Consider the determinantal point process on $\R_{>0}$ defined by the kernel $\mathbb{S}(x,y;g)$ of Theorem \ref{TheoremTransitionToBesselKernel} (b). We will show
that this process interpolates between the Bessel-kernel process, i.e. the determinantal process that has kernel $S^{\sqrt{\Bessel}}(x,y)$,
and the Meijer
G-kernel process, i.e. the determinantal point process that has kernel
$S_{\nu,0}^{\Ind}(x,y)$ defined by equation (\ref{SIndependent}).
This is stated in the next theorem.
\begin{thm}\label{TheoremTransition} Let the kernel
$\mathbb{S}(x,y;g)$
in the strongly correlated limit be defined by equation (\ref{Skernel}). Then for $x$, $y$
chosen from compact subsets
of $(0,+\infty)$ we have
\begin{equation}\label{TransitionKuijlaarsZhang} \underset{g\rightarrow+\infty}{\lim}\left[2g\mathbb{S}\left(2(gx)^{\frac{1}{2}},2(gy)^{\frac{1}{2}};g\right)\right]=S_{\nu,0}^{\Ind}(x,y).
\end{equation}
Furthermore,
\begin{equation}\label{deformation}
\underset{g\rightarrow 0}{\lim}\mathbb{S}\left(x,y;g\right)=S^{\sqrt{\Bessel}}(x,y).
\end{equation}
\end{thm}
Equation (\ref{deformation}) suggests to interpret the determinantal point process defined by the kernel  $\mathbb{S}(x,y;g)$ as a deformation of the Bessel-kernel process.
In such an interpretation $g$ plays a role of a deformation parameter.

One of the characteristics of the Bessel-kernel process is the one-point function $S^{\sqrt{\Bessel}}(x,x)$, also called microscopic density.
It is well know that $S^{\sqrt{\Bessel}}(x,x)$ can be represented in terms of Bessel functions, after applying l'H{\^o}pital's rule to the Bessel-kernel (\ref{BesselKernel}) at equal arguments, see, for example,
\cite[Sec. 7.2.1] {ForresterLogGases}.
We have
\begin{equation}\label{SBesselXX} S^{\sqrt{\Bessel}}(x,x)=\frac{1}{x}\left[\left(J_{\nu}\left(2x^{\frac{1}{2}}\right)\right)^2 -J_{\nu+1}\left(2x^{\frac{1}{2}}\right)J_{\nu-1}\left(2x^{\frac{1}{2}}\right)\right].
\end{equation}
Below we also give an explicit formula for the one-point function  $\mathbb{S}\left(x,x;g\right)$
characterising the density of particles of the deformed determinantal point process.
\begin{prop}\label{PropositionDensity1}
We have
\begin{equation}\label{SXX}
\begin{split}
\mathbb{S}(x,x;g)=&
\frac{1}{g^2x(2\pi i)^2}\oint\limits_{\Sigma}dt\oint\limits_{\Sigma}ds
\frac{\Gamma(-t)\Gamma(-s)x^{t+s+\nu}}{\Gamma(t+\nu+1)\Gamma(s+\nu+1)}\\
&\times\left[\mathcal{P}(s,t,\nu)-g(s^2+t^2+\nu s-st)\right]I_{t-1}\left(\frac{x}{2g}\right)K_{s+\nu}\left(\frac{x}{2g}\right),
\end{split}
\end{equation}
where $\mathcal{P}(s,t,\nu)$ is defined by equation (\ref{MainPolynomial2}), and $\Sigma$ is defined in the same way as in the statement of Theorem \ref{TheoremTransitionToBesselKernel} (b).
\end{prop}
Also, we will demonstrate independently that that the one-point function $\mathbb{S}(x,x;g)$ converges to $S^{\sqrt{\Bessel}}(x,x)$ as $g\rightarrow0$.
This is seen in the next proposition.
\begin{prop}\label{PropositionDensity2}
For fixed $x>0$ we have
\begin{equation}
\underset{g\rightarrow 0}{\lim}\ \mathbb{S}\left(x,x;g\right)
=
\frac{1}{(2\pi i)^2}\oint\limits_{\Sigma}dt\oint\limits_{\Sigma}ds
\frac{2t\Gamma(-t)\Gamma(-s)x^{t+s+\nu-3}}{\Gamma(t+\nu+1)\Gamma(s+\nu+1)}\mathcal{P}(s,t,\nu)=S^{\sqrt{\Bessel}}(x,x).
\end{equation}
\end{prop}
This limit to the one-point function of the Bessel process will be further discussed and illustrated in Section \ref{Comparison} below where we will also compare to Borodin's interpolating kernel.

\section{Comparing interpolating kernels}\label{Comparison}

In this section we will illustrate the interpolating process defined by the correlation kernel $\mathbb{S}\left(x,y;g\right)$
from Theorem \ref{TheoremTransitionToBesselKernel} (b) by plotting its density in the hard edge scaling limit from Proposition \ref{PropositionDensity1}, and by comparing it to the corresponding density in the
Muttalib-Borodin  (MB) ensemble \cite{Muttalib,Borodin:1998} which is another but different interpolating process.

We begin by briefly introducing  the results for this 
latter
ensemble,
for a very recent discussion of the Muttalib-Borodin ensemble we refer the reader to Forrester and Wang \cite{FW}.
Its joint density of points $y_1,\ldots,y_N>0$ is given by
\begin{equation}
\label{MBjpdf}
P_{\rm MB}(y_1,\ldots,y_N)=C \prod_{j=1}^N y_j^\alpha \e^{-y_j} \det\left[ y_k^{l-1}\right]_{k,l=1}^N \det\left[ y_k^{\theta(l-1)}\right]_{k,l=1}^N\ ,
\end{equation}
where $\alpha>-1$, $\theta>0$ and we suppress the known normalisation constant $C$. Clearly for $\theta=1$ this gives the joint density of the standard Laguerre ensemble (\ref{StandardLaguerrejpdf}) with $\nu=\alpha$.
The corresponding hard edge scaling limit of this ensemble is a determinantal point process on 
$(0,\infty)$ whose correlation kernel, $K^{(\alpha,\theta)}(x,y)$, is defined
for $x,y>0$
by
\begin{equation}
\label{BorodinsKernel}
K^{(\alpha,\theta)}(x,y)=\theta x^{\alpha}\int\limits_0^1du J_{{(\alpha+1)}{\theta}^{-1},{\theta}^{-1} }(xu)
J_{\alpha+1,\theta}\left((yu)^{\theta}\right)u^{\alpha},
\end{equation}
following
the notation of \cite{ArnoDries,FW}\footnote{In Borodin's notation \cite{Borodin:1998} it is defined without the prefactor $x^\alpha$.}.
Here
$J_{a,b}(x)$ is Wright's generalisation of the Bessel function,
$$
J_{a,b}(x)=\sum\limits_{j=0}^{\infty}\frac{(-x)^j}{j!\Gamma(a+bj)}\ ,\ \ x>0\ ,
$$
which satisfies $J_{a+1,b=1}(z^2/4)=(z/2)^{-a}J_{a}(z)$, after comparing to eq. (\ref{BesselJdef}).
The kernel $K^{(\alpha,\theta)}(x,y)$ was obtained 
by Borodin in 
\cite[Sec. 4]{Borodin:1998}.
It is not hard to see (
\cite[Example 3.5]{Borodin:1998}) that for $\theta=1$
$$ \left(\frac{y}{x}\right)^{\frac{\alpha}{2}}K^{(\alpha,\theta=1)}(x,y)=\frac{\varphi_1(x)\varphi_2(y)-\varphi_1(y)\varphi_2(x)}{x-y},
$$
where $\varphi_1(x)$ can be expressed in terms of the Bessel function of the first kind, $\varphi_1(x)=J_{\alpha}(2{x}^{\frac{1}{2}})$, and $\varphi_2(x)=x\varphi_1'(x)$.
Consequently for $\theta=1$ the correlation kernel $K^{(\alpha,\theta=1)}(x,y)$ is proportional to
the classical Bessel-kernel eq. (\ref{BesselKernel}), after changing variables accordingly as in eq. (\ref{change}). In particular, we have the following relation
\begin{equation}\label{MB-Bessel-map}
S^{\sqrt{\Bessel}}(x,y)=\left( \frac{y}{x}\right)^{{\nu}}(xy)^{-\frac{1}{2}} K^{(\nu,\theta=1)}(x,y).
\end{equation}
In \cite{ArnoDries}  Kuijlaars and Stivigny proved the following relation between Borodin's kernel and the Meijer G-kernel $K_{\nu_1,\ldots,\nu_M}(x,y)$:
\begin{thm}
\label{KS2014}
{\bf [Kuijlaars and Stivigny \cite{ArnoDries}]} Let $M\geq 1$ be an integer, then
\begin{equation}\label{RelationBorodinMeijer}
M^MK^{\left(\alpha,\frac{1}{M}\right)}\left(M^Mx,M^My\right)=\left(\frac{x}{y}\right)^{\alpha}K_{\nu_1,\ldots,\nu_M}(x,y),
\end{equation}
with parameters  $\nu_j=\alpha+\frac{j-1}{M},\;\;\; j=1,\ldots,M$ ,
and
\begin{equation}\label{RelationBorodinMeijer1}
x^{\frac{1}{M}-1}K^{\left(\alpha,{M}\right)}\left(Mx^{\frac{1}{M}},My^{\frac{1}{M}}\right)=K_{\tilde{\nu}_1,\ldots,\tilde{\nu}_M}(y,x),
\end{equation}
with parameters  $\tilde{\nu}_j=\frac{\alpha+j}{M}+1,\;\;\; j=1,\ldots,M$.
\end{thm}
Here the kernel $K_{\nu_1,\ldots,\nu_M}(x,y)$ is the Meijer G-kernel in the hard edge limit associated with the product of $M$
independent rectangular Gaussian matrices, and is defined in
\cite[Thm 5.3]{KuijlaarsZhang}, see eqs. (\ref{SIndependent}) and (\ref{SIndependentDefiniteIntegral}) for $M=2$.
Because of the range of parameters in the above theorem there is no choice of $\alpha$, for which Borodin's kernel becomes identical to that of the product of $M$ rectangular independent Gaussian matrices with integer values for all $\nu_j$ or all $\tilde{\nu}_j$. The only exception is $M=1$, in which case Borodin's kernel coincides with the Bessel-kernel.
It is in this sense that that Borodin's kernel $K^{(\alpha,\theta)}(x,y)$ interpolates between the Bessel-kernel at value $\theta=1$ on the one hand, and the Meijer G-kernel for $M=2$ and fractional values of $\nu_j$ and  $\tilde{\nu}_j$
specified in Theorem \ref{KS2014} at the values $\theta=\frac12$ and $\theta={2}$, respectively.
In contrast our limiting kernel $\mathbb{S}(x,y;g)$ interpolates between the Bessel-kernel in the limit $g\to0$ and the kernel for the product of $M=2$ independent rectangular Gaussian random matrices with integer parameter values $\nu_1=\nu$ and $\nu_2=0$ in the limit $g\to\infty$.

Furthermore, in \cite{FW} Corollary 5.3 a complex contour integral representation of the kernel eq. (\ref{BorodinsKernel}) was derived valid for arbitrary $\theta>0$:
\begin{equation}\label{KBrep}
K^{(\alpha,\theta)}(x,y)=
\frac{\theta}{(2\pi i)^2}\int\limits_{-\frac{1}{2}-i\infty}^{-\frac{1}{2}+i\infty}ds\int\limits_{\Sigma}dt\frac{\Gamma(s+1)\Gamma(s+\alpha+1)}{\Gamma(t+1)\Gamma(t+\alpha+1)}
\frac{\sin(\pi s)}{\sin(\pi t)}\frac{x^{-\theta s-1}y^{\theta t}}{(s-t)},
\end{equation}
where for $\theta<1$ the first contour has to be slightly modified by tilting it with respect to an angle $\delta$ when going towards $\pm i \infty$, see \cite{FW} for details.
For an alternative earlier contour integral representation valid at $M=2$ only see \cite{LSZ}.
Because the integral representation (\ref{KBrep}) is clearly different from the one in Theorem \ref{TheoremTransitionToBesselKernel} (b) and because of its two different limiting Meijer G-kernels as described above the kernels,  $K^{(\alpha,\theta)}(x,y)$ and $\mathbb{S}(x,y;g)$  clearly represent different interpolating point processes.

Let us illustrate the difference of these two interpolating point processes by plotting 
their respective densities. For that purpose it is advantageous to unfold these densities. A typical property of the density of the singular values is that 
its global density diverges at the origin. This divergence is also seen by taking the asymptotic limit of large argument of the local density. Unfolding then means to change variables such that the mean spacing becomes constant, and that the unfolded local density becomes asymptotically constant.

We begin with the known density of the Bessel-kernel.
Using the asymptotic expansion of the Bessel functions, see \cite{GradshteinRyzhik},
for $x>0$,
\begin{equation}
\label{JBesselasympt}
J_\nu(x)=\sqrt{\frac{2}{\pi x}}\cos\left(x-\frac{\pi}{4}(2\nu+1)\right)\left(1+O\left(\frac1x\right)
\right),
\end{equation}
as $x\rightarrow\infty$,
it is not difficult to obtain the asymptotics of the Bessel density eq. (\ref{SBesselXX})
\begin{equation}
\label{limSBessel}
S^{\sqrt{\Bessel}}(x,x)=\frac{x^{-\frac32}}{\pi}\left(1+O\left(\frac{1}{{x}^{\frac{1}{2}}}\right)\right) \ ,
\end{equation}
as $x\rightarrow\infty$.
We note that due to the change of variables, eq. (\ref{change}), the inverse square root divergence of the global Marchenko-Pastur law is mapped to $-3/2$. The unfolded Bessel density thus reads
\begin{equation}\label{rhomicroBessel}
\rho_{\rm micro}^{\Bessel}(x)=x^3 S^{\sqrt{\Bessel}}(x^2,x^2) = x \left (J_\nu(2x)^2-J_{\nu-1}(2x) J_{\nu+1}(2x)\right) \ ,
\end{equation}
which asymptotically becomes
$1/\pi$. This density is plotted in Figures \ref{FigInterpolating} and \ref{FigureMB} as a reference density compared to the two interpolating densities.

Next we turn to the asymptotics of the density of our interpolating kernel from Proposition \ref{PropositionDensity1}. Here we are only able to present heuristic arguments. Without giving much details we partly use the
proof of Proposition \ref{PropositionDensity2} from Section \ref{ProofInterpolatingDensity} where the asymptotic limit $g\to0$ is considered. There the asymptotic expansion of the two modified Bessel functions in eq. (\ref{SXX}) was determined in eqs. (\ref{KIasymptotic}) and (\ref{WXX}) and we obtain
\begin{eqnarray}
\mathbb{S}(x,x;g)&=&
\oint\limits_{\Sigma}dt\oint\limits_{\Sigma}ds
\frac{\Gamma(-t)\Gamma(-s)x^{t+s+\nu}\left[\mathcal{P}(s,t,\nu)-g(s^2+t^2+\nu s-st)\right]}{g^2x(2\pi i)^2\Gamma(t+\nu+1)\Gamma(s+\nu+1)}\nonumber\\
&&\times
\left(\frac{g}{x}+\frac{g^2}{x^2}\left((s+\nu)^2-(t-1)^2\right)
\right.\nonumber\\
&&\left.
+\frac{g^3}{x^3}\left(
\left((s+\nu)^2-(t-1)^2\right)^2 -2\left((s+\nu)^2+(t-1)^2\right)+1
\right)
+O\left(\frac{1}{x^4}\right)\right)\nonumber\\
&=&\frac{1}{x(2\pi i)^2}\oint\limits_{\Sigma}dt\oint\limits_{\Sigma}ds
\frac{\Gamma(-t)\Gamma(-s)x^{t+s+\nu}}{\Gamma(t+\nu+1)\Gamma(s+\nu+1)}
\left( \frac{1}{x^2}2t\mathcal{P}(s,t,\nu)
\right.\nonumber\\
&&\left.-\frac{g}{x^2}2t(s^2+t^2+\nu s-st)+O\left(\frac{1}{x^3}\right)
\right),
\label{SXXgasympt}
\end{eqnarray}
as $x\rightarrow\infty$.
\begin{figure}[h]
\centering
\epsfig{file=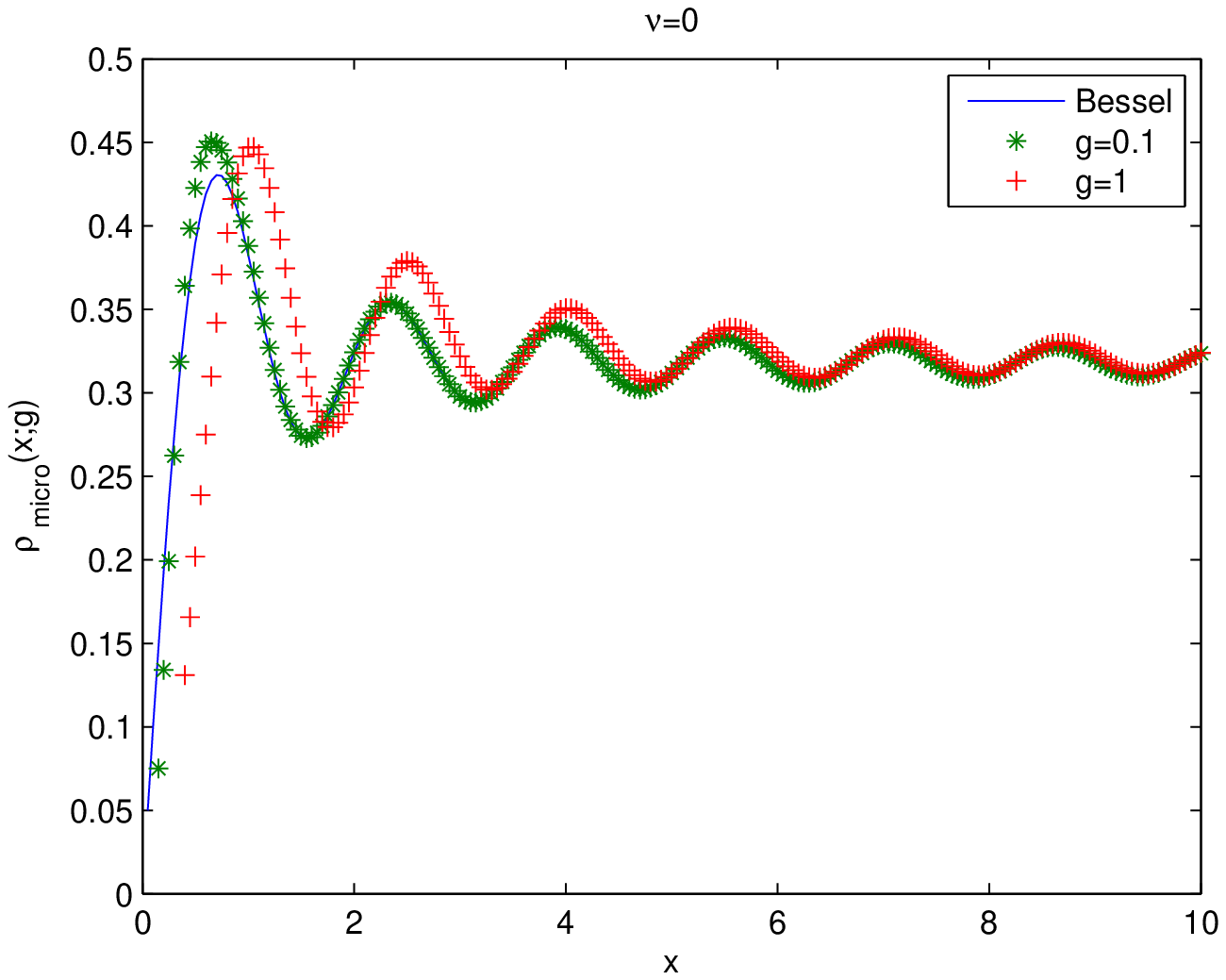,height=8cm,width=12cm}
\epsfig{file=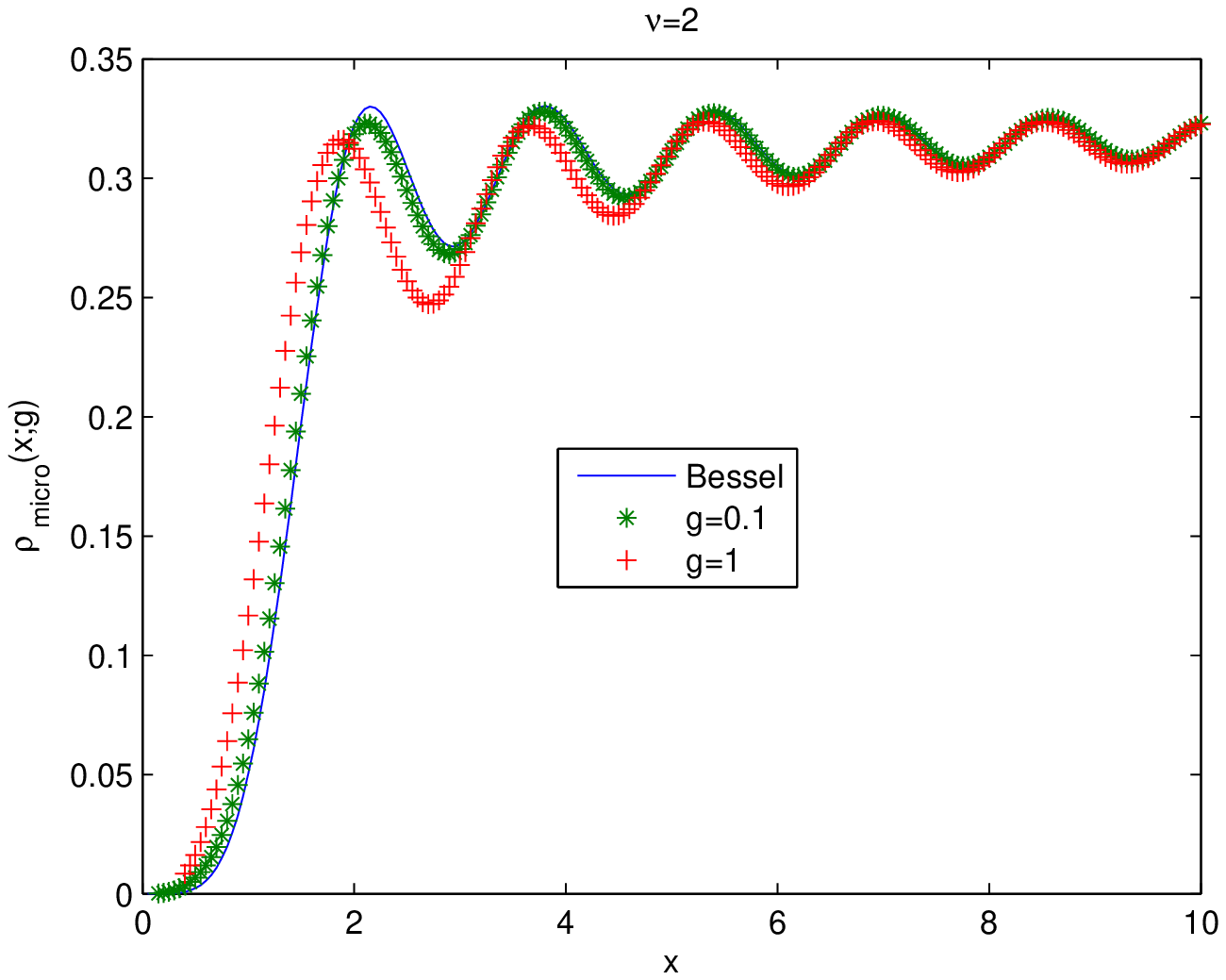,height=8cm,width=12cm}
\caption{We plot $\rho_{\rm micro}(x;g)$ from eq.
(\ref{rhomicroInterpolating})
for different values of $g=0.1$ and $g=1$,  together with
$\rho_{\rm micro}^{\Bessel}(x)$ from eq. (\ref{rhomicroBessel}) corresponding to $g=0$, at $\nu=0$ (upper plot) and $\nu=2$ (lower plot).
We see that for small enough $g$ the two densities
essentially coincide with each other.
}
\label{FigInterpolating}
\end{figure}
In the second step we have used the integral identities (\ref{Identity1})
- (\ref{Identity3}). We have been unable to prove that the
$O\left(x^{-3}\right)$ term does not contribute to the leading asymptotics of the integral. In the following we assume that this is the case, which is confirmed by numerics.

Inside the integral in eq. (\ref{SXXgasympt}) the term proportional to
$\mathcal{P}(s,t,\nu)$ was shown in eq. (\ref{Sfinallim}) to equal the Bessel density. The term proportional to ${g}/{x^2}$ can be expressed in term of Bessel functions using  Proposition \ref{PropositionB} and identities (\ref{Identity4}). An application of the expansion (\ref{JBesselasympt}) shows that it is subleading. From this we obtain the following conjectured asymptotics
\begin{equation}
\label{conjecture}
\mathbb{S}(x,x;g)=S^{\sqrt{\Bessel}}(x,x)\left(1+O\left(\frac{1}{{x}^{\frac{1}{2}}}\right)\right),
\end{equation}
as $x\rightarrow\infty$.
Consequently the unfolded interpolating density should be defined in the same way as the Bessel-density eq. (\ref{rhomicroBessel}), as it shares its asymptotic,
\begin{equation}\label{rhomicroInterpolating}
\rho_{\rm micro}(x;g)=x^3 {\mathbb S}(x^2,x^2;g) \ .
\end{equation}

The interpolating density unfolded in this way
is plotted in the Figure \ref{FigInterpolating} for several small values of $g>0$, compared to the Bessel-density  eq. (\ref{rhomicroBessel}). The close agreement with the Bessel-density for moderate $x$ at different values of $g>0$ further corroborates our conjectured asymptotics (\ref{conjecture}).

\begin{figure}[h]
\centering
\epsfig{file=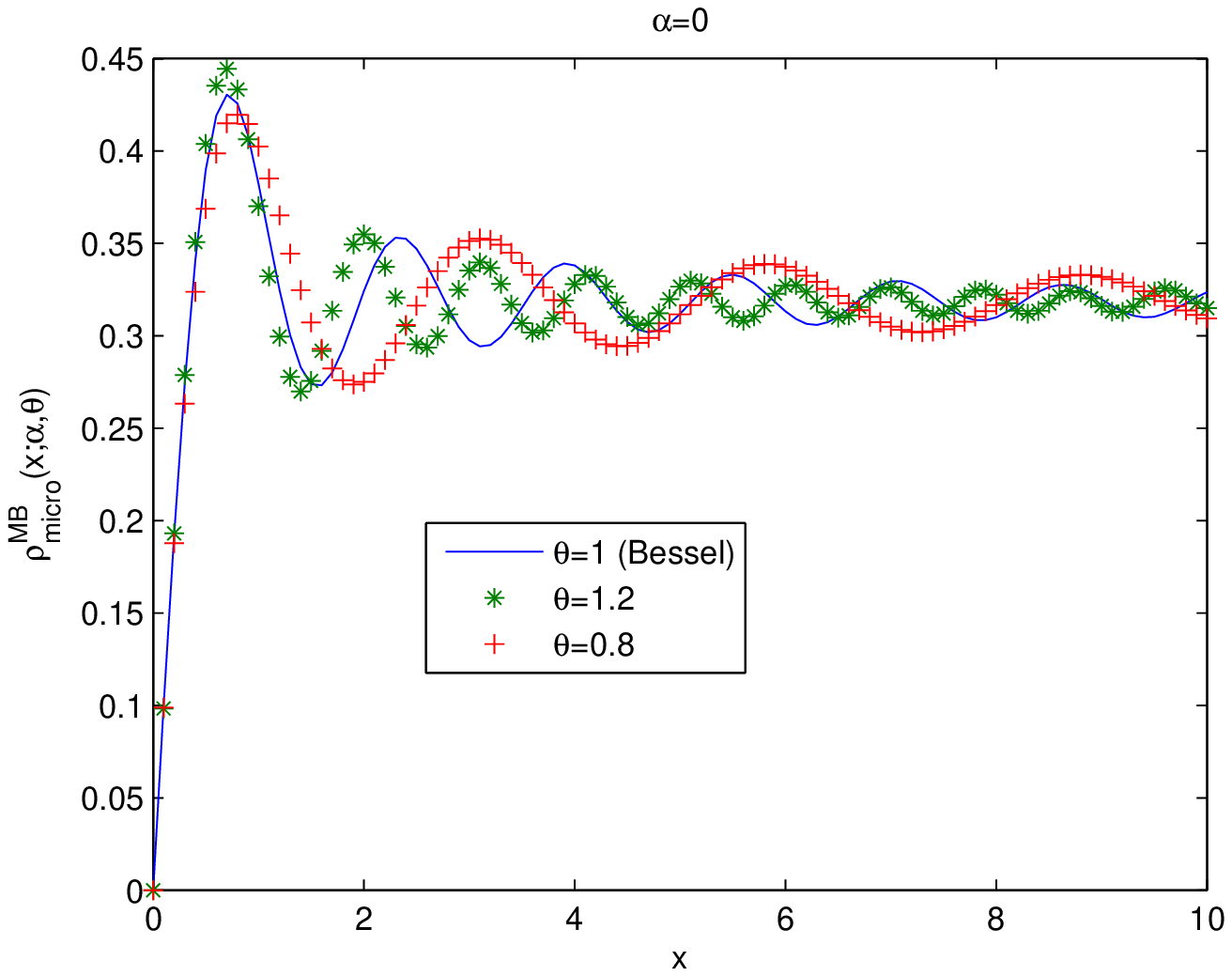,height=8cm,width=12cm}
\caption{This figure shows $\rho_{\rm micro}^{(\rm MB)}(x;\alpha,\theta)$ from eq. (\ref{rhomicroMB})
with $\alpha=0$, and the deformation parameter $\theta=1.2$, and $0.8$, compared to the Bessel-density $\rho_{\rm micro}^{\Bessel}(x)$ from eq. (\ref{rhomicroBessel}), which corresponds to $\theta=1$ at $\nu=0$.
}
\label{FigureMB}
\end{figure}
Finally we turn to the asymptotic expansion of the local density of the Muttalib-Borodin ensemble, using the integral representation (\ref{BorodinsKernel}). Here we use the asymptotic expansion of Wright's generalised Bessel functions from the original work \cite[Thm. 1]{EMW}
for $z>0$,
\begin{equation}\label{JWasymptot}
J_{a+1,\theta}(z)\sim \frac{\left(\theta z \e^{i\pi}\right)^{-{(2a+1)}/{(2(1+\theta))}}}{\sqrt{2\pi(1+\theta)}}
\exp\left[\frac{(1+\theta)}{\theta}\left(\theta z \e^{i\pi}\right)^{{1}/{(1+\theta)}}\right] \ +\ c.c.,
\end{equation}
as $z\rightarrow\infty$.
Here $c.c.$ stands for complex conjugate.
A tedious but straightforward application of this asymptotic expansion leads to the following result for $\theta>0$:
\begin{equation}\label{KMBasymptot}
K^{(\alpha,\theta)}(x,x)\sim \frac1\pi \theta^{{1}/{(1+\theta)}}\sin\left(\frac{\pi}{1+\theta}\right) x^{-{1}/{(1+\theta)}}\ ,
\end{equation}
as $x\rightarrow\infty$, which is independent of $\alpha$
to leading order.
We thus define the unfolded microscopic density as
\begin{equation}\label{rhomicroMB}
\rho_{\rm micro}^{(\rm MB)}(x;\alpha,\theta)=x K^{(\alpha,\theta)}\left(x^{\theta+1},x^{\theta+1}\right)\frac{\theta^{-{1}/{(1+\theta)}}}{\sin\left(\frac{\pi}{\theta+1}\right)} \ ,
\end{equation}
in order to normalise it to an asymptotic value of $1/\pi$.
It is plotted in Figure \ref{FigureMB}
for several values in the vicinity of the Bessel-kernel at $\theta=1$.

Comparing the Figures   \ref{FigInterpolating} and  \ref{FigureMB}
the two interpolating processes are clearly seen to be different.
Because the local maxima of the
microscopic density correspond to the average locations of the individual eigenvalues the difference in behaviour suggests the following interpretation. In the interpolating density $\rho_{\rm micro}(x;g)$ small deviations of the rescaled coupling of the matrices $g$ from zero mainly affect the repulsion from the origin, but not the spacing between eigenvalues further away from the origin. The means that for such a perturbation the eigenvalue repulsion basically remains as
in the unperturbed Laguerre ensemble (\ref{TheLaguerreProbabilityDensityFunction}) at $g=0$.

In contrast in the density $\rho_{\rm micro}^{(\rm MB)}(x;\alpha,\theta)$ the repulsion from the origin is basically unchanged for small deviations from $\theta=1$. However, the repulsion among eigenvalues is immediately increased for $\theta<1$ and reduced for $\theta>1$, as could be expected from the $\theta$-dependence inside the Vandermonde determinant in eq. (\ref{MBjpdf}).

\section{Proof of Theorem \ref{TheoremTransitionToBesselKernel}
}\label{Proof-Thm1.5}
Our aim  is to demonstrate the convergence of the kernel $K_N(x,y;\mu)$ defined by equation
(\ref{MainExactKernel}) in the different scaling limits (a), (b) and (c) using a contour integral representation for the correlation kernel.
For part (a) we need the asymptotic formulae for the Bessel functions
$I_t\left(\frac{1-\mu}{\mu}x^{\frac{1}{2}}\right)$ and $K_{s+\nu}\left(\frac{1+\mu}{\mu}y^{\frac{1}{2}}\right)$
to hold as $\mu\rightarrow 0$. Thus we begin with the following Proposition.
\begin{prop}\label{PropositionBesselLargeArgumentAsymtotics} As $\mu\rightarrow 0$,  the following asymptotic formulae hold true
for $s,t\in\mathbb{C}$ and $x,y>0$:
\begin{equation}
I_t\left(\frac{1-\mu}{\mu}x^{\frac{1}{2}}\right)\sim\frac{\mu^{\frac{1}{2}}}{(2\pi)^{\frac{1}{2}}}
\frac{\e^{\frac{x^{\frac{1}{2}}}{\mu}-x^{\frac{1}{2}}}}{x^{\frac{1}{4}}},\;\;
K_{s+\nu}\left(\frac{1+\mu}{\mu}y^{\frac{1}{2}}\right)\sim\frac{\pi^{\frac{1}{2}}\mu^{\frac{1}{2}}}{2^{\frac{1}{2}}}
\frac{\e^{-\frac{y^{\frac{1}{2}}}{\mu}-y^{\frac{1}{2}}}}{y^{\frac{1}{4}}}.
\label{IKasymptotic}
\end{equation}
\end{prop}
\begin{proof}
These asymptotic formulae can be obtained using the known asymptotic expansions of the Bessel functions $I_{\nu}(z)$ and $K_{\nu}(z)$
for large values of $|z|$
(we only need real positive here), see
\cite[Sec. 8.451]{GradshteinRyzhik}.
\end{proof}
Next, we need the asymptotics of the function $A_N(s,t;\mu)$ defined in (\ref{MainPolynomial}) which appears on the right-hand side
of formula (\ref{MainExactKernel}) for the correlation kernel $K_N(x,y;\mu)$.
\begin{prop}\label{PropositionAsymptoticsofA}
Assume that $\mu(N)=gN^{-\kappa}$, where $\kappa\geq1$
and $g$ is a constant which does not depend on $N$.\\
a) For $\kappa> 1$
$$
A_N(s,t;\mu(N))=\frac{1}{N}\mathcal{P}(s,t,\nu)+O\left(\frac{1}{N^\gamma}\right)\ , \ \ \gamma=\min\{\kappa,2\}\ .
$$
where $\mathcal{P}(s,t,\nu)$ is defined by equation (\ref{MainPolynomial2}).\\
b) For $\kappa=1$
\begin{equation}
\begin{split}
A_N(s,t;\mu(N))&=\frac{1}{N}\biggl\{\mathcal{P}(s,t,\nu)- g(s^2+t^2+\nu s-st)\biggr\}+O\left(\frac{1}{N^2}\right).
\end{split}
\end{equation}
\end{prop}
\begin{proof}
Looking at (\ref{MainPolynomial}) the  function $A_N(s,t;\mu)$ contains quadratic, linear and constant parts in $\mu$. It can thus be written as follows
\begin{equation}
\label{Asplit}
A_N(s,t;\mu)=(1+\mu^2)\alpha_N(s,t)+ \mu\beta_N(s,t)\ ,
\end{equation}
where $\alpha_N(s,t)=A_N(s,t;\mu=0)$ and $\beta_N(s,t)$ can be simply read off from eq. (\ref{MainPolynomial}). Clearly for any value of $\kappa\geq1$ the quadratic term
$O(\mu^2)$ is subleading. A simple calculation reveals the large-$N$ asymptotics for the two functions as
\begin{equation}\label{absplit}
\begin{split}
\alpha_N(s,t)&=\frac{1}{4N}(t-s)(t^2+s^2+(t+s)(\nu-1)-\nu)+O\left(\frac{1}{N^2}\right)= \frac{1}{N}\mathcal{P}(s,t,\nu)+O\left(\frac{1}{N^2}\right)\ ,\\
\beta_N(s,t)&=- (s^2+t^2+\nu s-st)+O\left(\frac{1}{N}\right)\ .
\end{split}
\end{equation}
Clearly the scaling of $\mu$ with $\kappa=1$ is special as only then the two
leading
terms on the right-hand side of eq. (\ref{absplit}) become of the same order
in eq. (\ref{Asplit}).
This concludes the proof.
\end{proof}
Next we turn to the contour integral representation of our kernel. Using formula (\ref{MainExactKernel}) for the correlation kernel $K_N(x,y;\mu)$ we write
\begin{equation}
\begin{split}
\frac{1}{N^2}&K_N\left(\frac{x}{N^2},\frac{y}{N^2};\mu(N)\right)\e^{-\frac{x^{\frac{1}{2}}}{N\mu(N)}+\frac{y^{\frac{1}{2}}}{N\mu(N)}}\\
&=\frac{1}{(2\pi i)^2}\frac{1}{x-y}\oint\limits_{\Sigma_N}dt\oint\limits_{\Sigma_N}ds
\frac{\Gamma(-t)\Gamma(-s)}{\Gamma(t+\nu+1)\Gamma(s+\nu+1)}\Phi_N\left(s,t;\mu(N)\right),
\end{split}
\nonumber
\end{equation}
where the function $\Phi_N\left(s,t;\mu(N)\right)$ is defined by
\begin{equation}\label{Phidef}
\begin{split}
\Phi_N\left(s,t;\mu(N)\right)=&\frac{1}{\mu(N)}I_t\left(\frac{1-\mu(N)}{\mu(N)}\frac{x^{\frac{1}{2}}}{N}\right)
K_{s+\nu}\left(\frac{1+\mu(N)}{\mu(N)}\frac{y^{\frac{1}{2}}}{N}\right)
\left(\frac{2}{1-\mu(N)}\frac{x^{\frac{1}{2}}}{N}\right)^t\\
\times&\left(\frac{2}{1+\mu(N)}\frac{y^{\frac{1}{2}}}{N}\right)^{s+\nu}
A_N(s,t;\mu(N)\frac{\Gamma(\nu+N+1)\Gamma(N+1)}{\Gamma(1-t+N)\Gamma(1-s+N)}\e^{-\frac{x^{\frac{1}{2}}}{N\mu(N)}+\frac{y^{\frac{1}{2}}}{N\mu(N)}}.
\end{split}
\nonumber
\end{equation}
As $N\rightarrow\infty$, we have the following asymptotics for the ratio of the Gamma-functions
\begin{equation}\label{GammaRatioAsymptotics}
\frac{\Gamma(\nu+N+1)\Gamma(N+1)}{\Gamma(1-t+N)\Gamma(1-s+N)}=N^{\nu+s+t}\left(1+O\left(N^{-1}\right)\right).
\end{equation}
\ \\

We will now complete part (a) of Theorem \ref{TheoremTransitionToBesselKernel}.
Assume that $\mu(N)=gN^{-\kappa}$, where
$\kappa>1$,
and $g$ is a constant.
Proposition \ref{PropositionBesselLargeArgumentAsymtotics}, Proposition \ref{PropositionAsymptoticsofA}, and  formula (\ref{GammaRatioAsymptotics})
imply that
$$
\underset{N\rightarrow\infty}{\lim}\Phi_N\left(s,t;\mu(N)\right)=\frac{\left(2x^{\frac{1}{2}}\right)^t\left(2y^{\frac{1}{2}}\right)^{s+\nu}\mathcal{P}(s,t,\nu)}{x^{\frac{1}{4}}y^{\frac{1}{4}}},
$$
where $\mathcal{P}(s,t,\nu)$ is defined by equation (\ref{MainPolynomial2}).
Define the contour $\Sigma$ as in eq. (\ref{SIndependent}) to start at $+\infty$ in the upper half plane, to encircle the positive real axis while keeping to the right of $-1/2$, and to return to $+\infty$ in the lower half plane. By modifying the contour $\Sigma_N$ to $\Sigma$ and taking the limit inside the integral we obtain
\begin{equation}\label{BesselKernelContourIntegralRepresentation}
\begin{split}
&\underset{N\rightarrow\infty}{\lim}\left(\frac{1}{N^2}K_N
\left(\frac{x}{N^2},\frac{y}{N^2};gN^{-\kappa}\right)\e^{-\frac{x^{\frac{1}{2}}}{gN^{1-\kappa}}+\frac{y^{\frac{1}{2}}}{gN^{1-\kappa}}}\right)\\
&=-\frac{1}{8(x-y)}\frac{1}{x^{\frac{1}{4}}y^{\frac{1}{4}}}
\frac{1}{(2\pi i)^2}\oint\limits_{\Sigma}dt\oint\limits_{\Sigma}ds
\frac{\Gamma(-t)\Gamma(-s)}{\Gamma(t+\nu+1)\Gamma(s+\nu+1)}\\
&\quad\times\biggl[s(s+\nu)(s+\nu-1)-t(t+\nu)(t+\nu-1)-\nu s(s+\nu)\\
&\quad\quad\quad+\nu t(t+\nu)+t(t+\nu)s-s(s+\nu)t\biggr]\left(2x^{\frac{1}{2}}\right)^t
\left(2y^{\frac{1}{2}}\right)^{s+\nu},
\end{split}
\end{equation}
where we have used the alternative form of the polynomial $\mathcal{P}(s,t,\nu)$ from eq. (\ref{MainPolynomial3}).
We note that  taking the limit inside the integral can be justified in the same way as in
\cite[Sec. 5.2]{KuijlaarsZhang}, using the dominated convergence theorem.

In order to express the right-hand side of the equation (\ref{BesselKernelContourIntegralRepresentation}) in terms of Bessel functions, we will use the following Proposition. We will state it in terms of the new variables $\xi$ and $\eta$
\begin{equation}
\label{newvariables}
\xi=2x^{\frac{1}{2}},\;\;\;\;\eta=2y^{\frac{1}{2}}.
\end{equation}

\begin{prop}\label{PropositionB}
We have for $\xi>0$ and with $\Sigma$ defined as in eq. (\ref{SIndependent}):
\begin{equation}
\begin{split}
\frac{1}{2\pi i}\oint\limits_{\Sigma}\frac{dt\Gamma(-t)}{\Gamma(\nu+1+t)}\xi^t=&
-\xi^{-\frac{\nu}{2}}J_{\nu}\left(2\xi^{\frac{1}{2}}\right),\\
\frac{1}{2\pi i}\oint\limits_{\Sigma}\frac{dt\,t\Gamma(-t)}{\Gamma(\nu+1+t)}\xi^t=&
\ \xi^{-\frac{\nu}{2}+\frac{1}{2}}J_{\nu+1}\left(2\xi^{\frac{1}{2}}\right),\\
\frac{1}{2\pi i}\oint\limits_{\Sigma}\frac{dt\,t(t+\nu)\Gamma(-t)}{\Gamma(\nu+1+t)}\xi^t=&
\ \xi^{-\frac{\nu}{2}+1}J_{\nu}\left(2\xi^{\frac{1}{2}}\right),\\
\frac{1}{2\pi i}\oint\limits_{\Sigma}\frac{dt\,t(t+\nu)(t+\nu-1)\Gamma(-t)}{\Gamma(\nu+1+t)}\xi^t=&
\ \xi^{-\frac{\nu}{2}+\frac{3}{2}}J_{\nu-1}\left(2\xi^{\frac{1}{2}}\right).
\end{split}
\end{equation}
\end{prop}
\begin{proof} Simple residue calculations using
$\underset{z=k}{\Res}\;\Gamma(-z)=-\frac{(-1)^{k}}{k!}$
which follows from the standard residue of $\Gamma(z)$ at $z=-k$
and the series expansion of the Bessel function
\begin{equation}
\label{BesselJdef}
J_\nu(z) = \sum_{m=0}^\infty \frac{(-1)^m}{\Gamma(m+1+\nu)}\left( \frac{z}{2}\right)^{2m+\nu}
\end{equation}
for $z>0$
leads to the desired equations.
\end{proof}

Applying  Proposition \ref{PropositionB} to the right-hand side of equation (\ref{BesselKernelContourIntegralRepresentation}) in terms of the variables eq. (\ref{newvariables}), we obtain
\begin{equation}
\label{Besselstep}
\begin{split}
&-\frac{1}{\xi^{\frac{1}{2}}\eta^{\frac{1}{2}}(\xi^2-\eta^2)}
\frac{1}{(2\pi i)^2}\oint\limits_{\Sigma}dt\oint\limits_{\Sigma}ds
\frac{\Gamma(-t)\Gamma(-s)}{\Gamma(t+\nu+1)\Gamma(s+\nu+1)}\\
&\times\biggl[s(s+\nu)(s+\nu-1)-t(t+\nu)(t+\nu-1)-\nu s(s+\nu)\\
&-\nu s(s+\nu)+\nu t(t+\nu)+t(t+\nu)s-s(s+\nu)t\biggr]\xi^t
\eta^{s+\nu}\\
=&\left(\frac{\eta}{\xi}\right)^{\frac{\nu}{2}}\biggl[
-J_{\nu}\left(2\xi^{\frac{1}{2}}\right)\eta^{\frac{3}{2}}J_{\nu-1}\left(2\eta^{\frac{1}{2}}\right)
+\xi^{\frac{3}{2}}J_{\nu-1}\left(2\xi^{\frac{1}{2}}\right)J_{\nu}\left(2\eta^{\frac{1}{2}}\right)\\
&+J_{\nu}\left(2\xi^{\frac{1}{2}}\right)\nu\eta J_{\nu}\left(2\eta^{\frac{1}{2}}\right)
-\nu\xi J_{\nu}\left(2\xi^{\frac{1}{2}}\right)J_{\nu}\left(2\eta^{\frac{1}{2}}\right)\\
&+\xi J_{\nu}\left(2\xi^{\frac{1}{2}}\right)\eta^{\frac{1}{2}}J_{\nu+1}\left(2\eta^{\frac{1}{2}}\right)
-\xi^{\frac{1}{2}}J_{\nu+1}\left(2\xi^{\frac{1}{2}}\right)\eta J_{\nu}\left(2\eta^{\frac{1}{2}}\right)\biggr]\left(-\frac{1}{\xi^{\frac{1}{2}}\eta^{\frac{1}{2}}(\xi^2-\eta^2)}\right)\\
=&\left(\frac{\eta}{\xi}\right)^{\frac{\nu}{2}}\biggl[
-J_{\nu}\left(2\xi^{\frac{1}{2}}\right)\eta^{\frac{3}{2}}J_{\nu-1}\left(2\eta^{\frac{1}{2}}\right)
+\xi^{\frac{3}{2}}J_{\nu-1}\left(2\xi^{\frac{1}{2}}\right)J_{\nu}\left(2\eta^{\frac{1}{2}}\right)\\
&+\xi^{\frac{1}{2}}J_{\nu-1}\left(2\xi^{\frac{1}{2}}\right)\eta J_{\nu}\left(2\eta^{\frac{1}{2}}\right)
-\xi J_{\nu}\left(2\xi^{\frac{1}{2}}\right)\eta^{\frac{1}{2}}J_{\nu-1}\left(2\eta^{\frac{1}{2}}\right)\biggr]\left(-\frac{1}{\xi^{\frac{1}{2}}\eta^{\frac{1}{2}}(\xi^2-\eta^2)}\right).
\end{split}
\end{equation}
In the last step we have used the recursion formula for the Bessel function, see e.g.
\cite{GradshteinRyzhik},
\begin{equation}
\label{BesselId}
\nu J_{\nu}\left(2\xi^{\frac{1}{2}}\right)=\xi^{\frac{1}{2}}J_{\nu-1}\left(2\xi^{\frac{1}{2}}\right)+\xi^{\frac{1}{2}}J_{\nu+1}\left(2\xi^{\frac{1}{2}}\right).
\end{equation}
The last two lines in equation (\ref{Besselstep}) can be simplified by simple factorisation, with $A=2\xi^{\frac12}$ and $B=2\eta^{\frac12}$
\begin{equation}
\begin{split}
&A^3J_{\nu-1}(A)J_{\nu}(B)-A^2BJ_{\nu-1}(B)J_{\nu}(A)+B^2AJ_{\nu-1}(A)J_{\nu}(B)-B^3J_{\nu-1}(B)J_{\nu}(A)\\
=&\left(A^2+B^2\right)\left[AJ_{\nu-1}(A)J_{\nu}(B)-BJ_{\nu-1}(B)J_{\nu}(A)\right].
\end{split}
\end{equation}
Putting all our results together we finally obtain
\begin{equation}\label{BesselContourIntegralRepresentation2}
\begin{split}
&\underset{N\rightarrow\infty}{\lim}
\left[\frac{1}{N^2}K_N\left(\frac{\xi^2}{4N^2},\frac{\eta^2}{4N^2};\frac{g}{N^{\kappa}}\right)\e^{-\frac{\xi}{2gN^{1-\kappa}}+\frac{\eta}{2gN^{1-\kappa}}}\right]\\
&=\left(\frac{\eta}{\xi}\right)^{\frac{\nu}{2}}\left(-\frac{1}{\xi^{\frac{1}{2}}\eta^{\frac{1}{2}}(\xi-\eta)}\right)
\left[2\xi^{\frac{1}{2}}J_{\nu-1}\left(2\xi^{\frac{1}{2}}\right)J_{\nu}\left(2\eta^{\frac{1}{2}}\right)
-2\eta^{\frac{1}{2}}J_{\nu-1}\left(2\eta^{\frac{1}{2}}\right)
J_{\nu}\left(2\xi^{\frac{1}{2}}\right)\right]\\
&=S^{\sqrt{\Bessel}}(\xi,\eta)\ ,
\end{split}
\end{equation}
where $S^{\sqrt{\Bessel}}(\xi,\eta)$ was defined by equation (\ref{BesselKernel}).
Thus Theorem \ref{TheoremTransitionToBesselKernel}, (a) is proved.
\qed

The proof of Theorem \ref{TheoremTransitionToBesselKernel}, (b) is rather similar, except that we do not
apply Proposition \ref{PropositionBesselLargeArgumentAsymtotics}. Here the arguments of the Bessel functions remain finite when scaling $\mu(N)=gN^{-1}$ in eq. (\ref{MainExactKernel}), where $g>0$ is a constant.
We start from the double contour integral representation for the correlation kernel $K_N(x,y;\mu)$ equation (\ref{MainExactKernel}).
Then we use Proposition
\ref{PropositionAsymptoticsofA}, b) for $\kappa=1$ and equation (\ref{GammaRatioAsymptotics}) for the asymptotics of the Gamma-functions to conclude in the same way that equation (\ref{BesselKernelContourIntegralRepresentation})
is replaced by
\begin{equation}
\begin{split}
\underset{N\rightarrow\infty}{\lim}
\left[\frac{1}{N^2}K_N\left(\frac{x}{4N^2},\frac{y}{4N^2};\frac{g}{N}\right)\right]
=
\frac{4}{(x-y)g}
\frac{1}{(2\pi i)^2}\oint\limits_{\Sigma}dt\oint\limits_{\Sigma}ds
\frac{\Gamma(-t)\Gamma(-s)x^{\frac{t}{2}}y^{\frac{s+\nu}{2}}}{\Gamma(t+\nu+1)\Gamma(s+\nu+1)}&\\
\times\left[\mathcal{P}(s,t,\nu)-g(s^2+t^2+\nu s-st)\right]I_t\left(\frac{x^{\frac{1}{2}}}{2g}\right)K_{s+\nu}\left(\frac{y^{\frac{1}{2}}}{2g}\right)&.
\end{split}
\end{equation}
This formula is equivalent to that in the statement of Theorem \ref{TheoremTransitionToBesselKernel}, (b).
\qed

For the proof of part (c) of Theorem \ref{TheoremTransitionToBesselKernel} we  refer the reader to the proof of Theorem 3.9 in \cite{AkemannStrahov}.
Starting from the
representation of  the correlation kernel eq. (10.1) in \cite{AkemannStrahov} one can repeat all calculations identically. The only difference
is to demonstrate that the term proportional to $B(s,t;N)$ in eq. (10.2) in \cite{AkemannStrahov} still vanishes in the large-$N$ limit in this scaling regime.
It is not hard to check  that $B(s,t; N)$ is of order ${O}(N^{-1})$. Together with $\mu(N)=gN^{-\kappa}$ for $1>\kappa\geq0$ this is sufficient to reach the same conclusions as in \cite{AkemannStrahov} for $\kappa=0$, that the limiting kernel is given by eq. (\ref{SIndependentDefiniteIntegral}).
\qed

\section{Proof of Theorem \ref{TheoremTransition}}\label{ProofInterpolatingKernel}
In this section we consider the new limiting kernel $\mathbb{S}(x,y;g)$ obtained in Theorem  \ref{TheoremTransitionToBesselKernel}, (b). The kernel $\mathbb{S}(x,y;g)$ is defined explicitly by equation
(\ref{Skernel}). Our aim is to show that the kernel $\mathbb{S}(x,y;g)$ defines a determinantal point process
interpolating between the Bessel-kernel process, and the point process describing the hard edge
scaling limit for the product of two independent Gaussian complex matrices.  The  Bessel-kernel process is defined by the kernel $S^{\sqrt{\Bessel}}(x,y)$,
and the  point process describing the hard edge
scaling limit for the product of two independent Gaussian complex matrices that we consider is defined by the kernel $S_{\nu,0}^{\Ind}(x,y)$. Thus it is enough to show that equations
(\ref{TransitionKuijlaarsZhang}) and (\ref{deformation}) hold true.

Let us first prove equation (\ref{TransitionKuijlaarsZhang}). We use equation (\ref{SIntegrableForm}) for the kernel $\mathbb{S}(x,y;g)$, and
note that in order to investigate the asymptotics of the kernel $\mathbb{S}(x,y;g)$ as $g\rightarrow+\infty$ it is convenient to use different (from those given by equations (\ref{FunctionF1})-(\ref{FunctionF4})) integral representations of the functions  $F_1(y;g)$, $F_2(y;g)$, $F_3(y;g)$ and $F_4(y;g)$.
\begin{prop}\label{Prop4.1}
The following formulae hold true
with $c>0$
\begin{equation}
F_1(y;g)=-\frac{(4g)^{\nu}}{2\Gamma(1+\nu)}\frac{1}{2\pi i}\int\limits_{c-i\infty}^{c+i\infty}ds\Gamma(s+\nu)\Gamma(s)
\Phi\left(\begin{array}{c}
             s+\nu \\
             \nu+1
           \end{array}
\biggl|-4g\right)\left(\frac{y}{4g}\right)^{-2s},
\end{equation}
\begin{equation}
F_2(y;g)=\frac{(4g)^{\nu+1}}{2\Gamma(\nu+2)}\frac{1}{2\pi i}
\int\limits_{c-i\infty}^{c+i\infty}ds\Gamma(s+\nu+1)\Gamma(s)
\Phi\left(\begin{array}{c}
             s+\nu+1 \\
             \nu+2
           \end{array}
\biggl|-4g\right)\left(\frac{y}{4g}\right)^{-2s},
\end{equation}
\begin{equation}
F_3(y;g)=\frac{(4g)^{\nu+1}}{2\Gamma(\nu+1)}\frac{1}{2\pi i}\int\limits_{c-i\infty}^{c+i\infty}ds\Gamma(s+\nu+1)\Gamma(s)
 \Phi\left(\begin{array}{c}
             s+\nu+1 \\
             \nu+1
           \end{array}
\biggl|-4g\right)\left(\frac{y}{4g}\right)^{-2s},
\end{equation}
and
\begin{equation}
F_4(y;g)=\frac{(4g)^{\nu+1}}{2\Gamma(\nu)}\frac{1}{2\pi i}\int\limits_{c-i\infty}^{c+i\infty}ds\Gamma(s+\nu+1)\Gamma(s)
\Phi\left(\begin{array}{c}
             s+\nu+1 \\
             \nu
           \end{array}
\biggl|-4g\right)\left(\frac{y}{4g}\right)^{-2s}.
\end{equation}
Here $c>0$, and $\Phi\left(\begin{array}{c}
             a \\
             b
           \end{array}
\biggl|x\right)$ is the confluent hypergeometric function, see \cite{NIST}, 
$$
\Phi\left(\begin{array}{c}
             a \\
             b
           \end{array}
\biggl|x\right)=\sum\limits_{k=0}^{\infty}\frac{(a)_k}{(b)_k}\frac{x^k}{k!}.
$$
\end{prop}
\begin{proof}
Let us prove the formula for $F_1(y;g)$. By the Residue Theorem
$\underset{z=k}{\Res}\;\Gamma(-z)=-\frac{(-1)^{k}}{k!}$, and
we can write
$$
F_1(y;g)=-\sum\limits_{k=0}^{\infty}\frac{(-1)^k}{k!}\frac{y^{k+\nu}}{\Gamma(k+\nu+1)}K_{k+\nu}\left(\frac{y}{2g}\right).
$$
Note that \cite[Sec. 9.34]{GradshteinRyzhik}
$$
K_{k+\nu}\left(\frac{y}{2g}\right)=\frac{1}{2}\left(\frac{y}{4g}\right)^{-k-\nu}G^{2,0}_{0,2}\left(\begin{array}{cc}
- &  \\
k+\nu & 0
\end{array}
\biggl|\frac{y^2}{16g^2}\right).
$$
Then we use the formula
from \cite{NIST} with $c>0$,
$$
 G^{2,0}_{0,2}\left(\begin{array}{cc}
- &  \\
b_1 & b_2
\end{array}
 \biggl|x\right)=\frac{1}{2\pi i}\int\limits_{c-i\infty}^{c+i\infty}\Gamma(s+b_1)\Gamma(s+b_2)x^{-s}ds,
$$
and the interchange the sum and the integral.
The justification for this is given in Appendix \ref{AppB}.
As a result we obtain the formula for $F_1(y; g)$ stated in Proposition \ref{Prop4.1}. The formulae for $F_1(y; g)$, $F_2(y; g)$, $F_3(y;g)$,
and $F_4(y;g)$ can be derived in the same way.
\end{proof}
Next, we give  integral representations  for the  functions $\Phi_1(x;g)$, $\Phi_2(x;g)$, $\Phi_3(x;g)$, and $\Phi_4(x;g)$ from eqs. (\ref{Phi1})-- (\ref{Phi4}), and
$F_1(y; g)$, $F_2(y; g)$, $F_3(y; g)$,
and $F_4(y; g)$ from above valid asymptotically as $g\rightarrow+\infty$.
\begin{prop}\label{PropositionFAsymptotics} As $g\rightarrow+\infty$,
$$
\Phi_1(x;g)\sim\frac{1}{2\pi i}\oint\limits_{\Sigma}dt\frac{\Gamma(-t)}{\Gamma(t+\nu+1)\Gamma(t+1)}\left(\frac{x^2}{4g}\right)^t,
$$
$$
\Phi_2(x;g)\sim\frac{1}{2\pi i}\oint\limits_{\Sigma}dt\frac{t\Gamma(-t)}{\Gamma(t+\nu+1)\Gamma(t+1)}\left(\frac{x^2}{4g}\right)^t,
$$
$$
\Phi_3(x;g)\sim\frac{1}{2\pi i}\oint\limits_{\Sigma}dt\frac{t(t+\nu)\Gamma(-t)}{\Gamma(t+\nu+1)\Gamma(t+1)}\left(\frac{x^2}{4g}\right)^t,
$$
and
$$
\Phi_4(x;g)\sim\frac{1}{2\pi i}\oint\limits_{\Sigma}dt\frac{t(t+\nu)(t+\nu-1)\Gamma(-t)}{\Gamma(t+\nu+1)\Gamma(t+1)}\left(\frac{x^2}{4g}\right)^t,
$$
Moreover, as $g\rightarrow+\infty$,
$$
F_1(y;g)\sim-\frac{1}{2}\frac{1}{2\pi i}\int\limits_{c-i\infty}^{c+i\infty}ds\frac{\Gamma(s+\nu)\Gamma(s)}{\Gamma(1-s)}\left(\frac{y^2}{4g}\right)^{-s},
$$
$$
F_2(y;g)\sim\frac{1}{2}\frac{1}{2\pi i}\int\limits_{c-i\infty}^{c+i\infty}ds\frac{\Gamma(s+\nu+1)\Gamma(s)}{\Gamma(1-s)}\left(\frac{y^2}{4g}\right)^{-s},
$$
$$
F_3(y;g)\sim-\frac{1}{2}\frac{1}{2\pi i}\int\limits_{c-i\infty}^{c+i\infty}ds\frac{s\Gamma(s+\nu+1)\Gamma(s)}{\Gamma(1-s)}\left(\frac{y^2}{4g}\right)^{-s},
$$
and
$$
F_4(y;g)\sim\frac{1}{2}\frac{1}{2\pi i}\int\limits_{c-i\infty}^{c+i\infty}ds\frac{s(s+1)\Gamma(s+\nu+1)\Gamma(s)}{\Gamma(1-s)}\left(\frac{y^2}{4g}\right)^{-s}.
$$
\end{prop}
\begin{proof}
The asymptotic formulae for the functions $\Phi_1(x;g)$, $\Phi_2(x;g)$, $\Phi_3(x;g)$, and $\Phi_4(x;g)$ stated in Proposition \ref{PropositionFAsymptotics} can be obtained replacing $I_t\left(\frac{x}{2g}\right)$
in equations (\ref{Phi1})-(\ref{Phi4}) by the first term of its series representation eq. (\ref{BesselFunctionI}). 
Let us derive in detail the asymptotic formula for $\Phi_1(x;g)$. Using the Residue Theorem we rewrite
$\Phi_1(x;g)$ as 
$$
\Phi_1(x;g)=-\sum\limits_{k=0}^{\infty}\frac{(-1)^kx^k}{k!\Gamma(k+\nu+1)}I_{k}\left(\frac{x}{2g}\right).
$$
Inserting the series representation eq.  (\ref{BesselFunctionI}) for  $I_{k}\left(\frac{x}{2g}\right)$ we get
\begin{equation}
\begin{split}
&\Phi_1(x;g)=
-\sum\limits_{k=0}^{\infty}\frac{(-1)^kx^{2k}}{k!\Gamma(k+\nu+1)(4g)^k}\left(\sum\limits_{m=0}^{\infty}\frac{1}{\Gamma(k+m+1)\,m!}\left(\frac{x}{4g}\right)^{2m}\right)
=\varphi(x;g)+\epsilon(x;g),
\end{split}
\end{equation}
where
\begin{equation}\label{0v}
\varphi(x;g)=-\sum\limits_{k=0}^{\infty}\frac{(-1)^kx^{2k}}{(k!)^2\Gamma(k+\nu+1)(4g)^k},
\end{equation}
is the term with $m=0$, and
$$
\epsilon(x;g)=-\sum\limits_{k=0}^{\infty}\frac{(-1)^kx^{2k}}{k!\Gamma(k+\nu+1)}\left(\sum\limits_{m=1}^{\infty}\frac{1}{\Gamma(k+m+1)\, m!}\left(\frac{x}{4g}\right)^{2m}\right).
$$
Note that the function $\varphi(x;g)$ can be represented as
$$
\varphi(x;g)=\frac{1}{2\pi i}\oint\limits_{\Sigma}dt\frac{\Gamma(-t)}{\Gamma(t+\nu+1)\Gamma(t+1)}\left(\frac{x^2}{4g}\right)^t,
$$
i.e. $\varphi(x;g)$ coincides with the right-hand side of the asymptotic formula for $\Phi_1(x;g)$ in the statement of Proposition \ref{PropositionFAsymptotics}.
We conclude that it is enough to check that
\begin{equation}\label{0000}
\epsilon(x;g)=o\left(\varphi(x;g)\right), 
\end{equation}
as $g\rightarrow+\infty$. To show this we take into account that the series
$$
\sum\limits_{k=0}^{\infty}\frac{(-1)^k}{(k!)^2\Gamma(k+\nu+1)}z^k
$$
converges uniformly in any closed ball with its center at $0$. Therefore the series on the right-hand of equation (\ref{0v}) 
defines a continuous function of $\frac{x^2}{4g}$. This gives
that for any fixed $x>0$ 
\begin{equation}\label{limitphi}
\underset{g\rightarrow+\infty}{\lim}\varphi(x;g)=-\frac{1}{\Gamma(\nu+1)}.
\end{equation}
On the other hand, we have
\begin{equation}\label{iii}
\left|\epsilon(x;g)\right|\leq\left(\sum\limits_{k=0}^{\infty}\frac{x^{2k}}{k!\Gamma(k+\nu+1)(4g)^k}\right)
\left(\sum\limits_{m=1}^{\infty}\frac{1}{(m!)^2}\left(\frac{x}{4g}\right)^{2m}\right).
\end{equation}
Both series 
$$
\sum\limits_{k=0}^{\infty}\frac{1}{k!\Gamma(k+\nu+1)}z^k
$$
and
$$
\sum\limits_{m=1}^{\infty}\frac{1}{(m!)^2}z^k
$$
converge uniformly in any closed ball with its center at $0$. Therefore, both two series on the right-hand side
of inequality (\ref{iii}) define continuous functions of $\frac{x^2}{4g}$ and $\frac{x}{4g}$ correspondingly, so 
\begin{equation}\label{limite}
\underset{g\rightarrow+\infty}{\lim}\epsilon(x;g)=0.
\end{equation}
From the limiting relations (\ref{limitphi}), (\ref{limite}) we conclude that (\ref{0000}) holds true.
The asymptotic formulae  for $\Phi_2(x;g)$, $\Phi_3(x;g)$ and $\Phi_4(x;g)$ follow in the same way.

To derive the asymptotic formulae for the functions
$F_1(y; g)$, $F_2(y; g)$, $F_3(y; g)$,
and $F_4(y; g)$ from Proposition \ref{Prop4.1} 
use the asymptotic formula for the confluent hypergeometric function as $\re(x)\rightarrow-\infty$, namely \cite{NIST}
$$
\Phi\left(\begin{array}{c}
             a \\
             b
           \end{array}
\biggl|x\right)=\frac{\Gamma(b)}{\Gamma(b-a)}\left(-x\right)^{-a}\left(1+O\left(|x|^{-1}\right)\right),\;\;\;\re (x)\rightarrow-\infty.
$$
\end{proof}

Now we are ready to show that equation (\ref{TransitionKuijlaarsZhang}) holds true.
We use equation (\ref{SIntegrableForm}) for the correlation kernel $\mathbb{S}(x,y;g)$,
which represents $\mathbb{S}(x,y;g)$ as a sum of two terms. Considering the scaling limit
of $\mathbb{S}(x,y;g)$ as on the left-hand side of equation (\ref{TransitionKuijlaarsZhang}) (and taking into account the asymptotic expressions obtained
in Proposition \ref{PropositionFAsymptotics})
we see that the first term (= first two lines) on the right-hand side
of equation (\ref{SIntegrableForm}) gives zero contribution. Proposition \ref{PropositionFAsymptotics} implies that the second term (= third and last line) on the right-hand side of equation
(\ref{SIntegrableForm}) is asymptotically equal to
\begin{equation}
\begin{split}
&\frac{2}{(x^2-y^2)}
\frac{1}{(2\pi i)^2}\oint\limits_{\Sigma}dt\int\limits_{c-i\infty}^{c+i\infty}ds
\frac{\Gamma(-t)}{\Gamma(t+\nu+1)\Gamma(t)}\frac{\Gamma(s+\nu)\Gamma(s)}{\Gamma(1-s)}\\
&\times\left[t(t+\nu)+s(s+\nu)+ts\right]\left(\frac{x^2}{4g}\right)^t\left(\frac{y^2}{4g}\right)^{-s},
\end{split}
\nonumber
\end{equation}
as $g\rightarrow+\infty$. This implies the limiting relation
\begin{equation}\label{ScalingLimitasDoubleIntegral}
\begin{split}
&\underset{g\rightarrow+\infty}{\lim}\left[2g\mathbb{S}\left(2(gx)^{\frac{1}{2}},2(gy)^{\frac{1}{2}};g\right)\right]\\
&=\frac{1}{(x-y)}
\frac{1}{(2\pi i)^2}\oint\limits_{\Sigma}dt\int\limits_{c-i\infty}^{c+i\infty}ds
\frac{\Gamma(-t)}{\Gamma(t+\nu+1)\Gamma(t)}\frac{\Gamma(s+\nu)\Gamma(s)}{\Gamma(1-s)}\\
&\quad\times\left[t(t+\nu)+s(s+\nu)+ts\right]x^ty^{-s}.
\end{split}
\end{equation}
As it is shown by the authors in
\cite[Sec.10]{AkemannStrahov},  the right-hand side of equation (\ref{ScalingLimitasDoubleIntegral}) can be rewritten as
$$
\int\limits_0^1G^{1,0}_{0,3}\left(\begin{array}{ccc}
                       & - &  \\
                      0, & -\nu, & 0
                    \end{array}
\biggr|ux\right)G^{2,0}_{0,3}\left(\begin{array}{ccc}
                       & - &  \\
                      \nu,&0, & 0
                    \end{array}
\biggr|uy\right)du.
$$
Now, equation (\ref{SIndependentDefiniteIntegral}) says that the right-hand side of equation (\ref{ScalingLimitasDoubleIntegral})
is equal to $S_{\nu,0}^{\Ind}(x,y)$,
and
equation (\ref{TransitionKuijlaarsZhang})
follows.

Finally, let us show that equation (\ref{deformation}) holds true. For this purpose we use the asymptotic expansions of the Bessel functions $I_{\nu}(z)$
and $K_{\nu}(z)$ for large values of $|z|$, see eqs. (\ref{KIasymptotic}) below, and find that as $g\rightarrow 0$, the right-hand side of equation (\ref{Skernel})
turns into the right-hand side of equation
(\ref{Besselstep}) and thus equation (\ref{BesselContourIntegralRepresentation2}) (where $\xi$ is replaced by $x$, and $\eta$ is replaced by $y$).
This observation gives equation (\ref{deformation}). Theorem \ref{TheoremTransition} is proved.
\qed

\section{Proof of Proposition \ref{PropositionDensity1} and Proposition \ref{PropositionDensity2}}\label{ProofInterpolatingDensity}

Equation (\ref{SXX}) for the kernel at equal arguments
$\mathbb{S}(x,x;g)$  which gives the one-point function
or microscopic density
can be obtained from formula (\ref{Skernel}).
It is clear that the limiting kernel at equal arguments is regular, and that hence the integral in eq.  (\ref{Skernel}) vanishes at equal arguments in order to cancel the pole at $x=y$ in front of the integrand,
\begin{equation}\label{SkernelIntegrand}
\begin{split}
0=&
\frac{1}{(2\pi i)^2}\oint\limits_{\Sigma}dt\oint\limits_{\Sigma}ds
\frac{\Gamma(-t)\Gamma(-s)x^{t+s+\nu}}{\Gamma(t+\nu+1)\Gamma(s+\nu+1)}\\
&\times\left[\mathcal{P}(s,t,\nu)-g(s^2+t^2+\nu s-st)\right]I_t\left(\frac{x}{2g}\right)K_{s+\nu}\left(\frac{x}{2g}\right).
\end{split}
\end{equation}
Taking into account the identity  for Bessel functions \cite{GradshteinRyzhik}
$$
\frac{\partial}{\partial x}\left(x^tI_t\left(\frac{x}{2g}\right)\right)=\frac{x^t}{2g}I_{t-1}\left(\frac{x}{2g}\right),
$$
as well as eq. (\ref{SkernelIntegrand}), the application of l'H$\hat{\mbox{o}}$pital's rule to the right-hand side of equation (\ref{Skernel}) yields formula (\ref{SXX}).
Alternatively  eq. (\ref{SkernelIntegrand}) can be shown to directly result from the limit of eq. (\ref{Integrand}) and
eq.  (\ref{SXX}) as the limit of eq. (\ref{MainExactKernelyy}), following Section \ref{Proof-Thm1.5}.
Proposition \ref{PropositionDensity1} is proved.\qed

Our next aim is to show that the one-point correlation function $\mathbb{S}(x,x;g)$ defined by equation (\ref{SXX})
converges to the one-point correlation function $S^{\sqrt{\Bessel}}(x,x)$ defined by equation (\ref{SBesselXX}) as $g\rightarrow 0$.
Starting from  equation (\ref{SXX}) we need
the asymptotic expansion of the modified Bessel functions of first and second kind \cite{GradshteinRyzhik} for $z>0$
\begin{eqnarray}\label{KIasymptotic}
I_{t}(z)&\sim&\frac{\e^z}{\sqrt{2\pi} {z}^{\frac{1}{2}}}\left(1-\frac{\Gamma(t+\frac{3}{2})}{2z\Gamma(t-\frac{1}{2})}
+\frac{\Gamma(t+\frac{5}{2})}{2(2z)^2\Gamma(t-\frac{3}{2})}
+O\left(z^{-2}\right)\right),\nonumber\\
K_{s+\nu}(z)&\sim&\sqrt{\frac{\pi}{2}}\frac{\e^{-z}}{{z}^{\frac{1}{2}}}
\left(1+\frac{\Gamma(s+\nu+\frac{3}{2})}{2z\Gamma(s+\nu-\frac{1}{2})}
+\frac{\Gamma(s+\nu+\frac{5}{2})}{2(2z)^2\Gamma(s+\nu-\frac{3}{2})}
+O\left(z^{-2}\right)\right),
\end{eqnarray}
as $z\rightarrow+\infty$. This gives in particular
\begin{eqnarray}\label{WXX}
I_{t}\left(\frac{x}{2g}\right)K_{s+\nu}\left(\frac{x}{2g}\right)
&=&
\frac{g}{x}-\frac{g^2}{x^2}\left(
\frac{\Gamma(t+\frac{3}{2})}{\Gamma(t-\frac{1}{2})}
-\frac{\Gamma(s+\nu+\frac{3}{2})}{\Gamma(s+\nu-\frac{1}{2})}\right)\nonumber\\
&&+\frac{g^3}{x^3}\left(
\frac{\Gamma(s+\nu+\frac{5}{2})}{2\Gamma(s+\nu-\frac{3}{2})}
+\frac{\Gamma(t+\frac{5}{2})}{2\Gamma(t-\frac{3}{2})}
-\frac{\Gamma(s+\nu+\frac{3}{2})\Gamma(t+\frac{3}{2})}{\Gamma(s+\nu-\frac{1}{2})\Gamma(t-\frac{1}{2})}
\right)\nonumber\\
&&+O\left(\frac{g^4}{x^4}\right)\nonumber\\
&=&\frac{g}{x}+\frac{g^2}{x^2}\left((s+\nu)^2-t^2\right)\nonumber\\
&&+\frac{g^3}{x^3}\left(
\left((s+\nu)^2-t^2\right)^2 -2\left((s+\nu)^2+t^2\right)+1
\right)
+O\left(\frac{g^4}{x^4}\right).
\end{eqnarray}
Applying this asymptotic expansion to eq. (\ref{SkernelIntegrand})
we obtain
\begin{equation}
\begin{split}
0=&
\frac{1}{(2\pi i)^2}\oint\limits_{\Sigma}dt\oint\limits_{\Sigma}ds
\frac{\Gamma(-t)\Gamma(-s)x^{t+s+\nu}}{\Gamma(t+\nu+1)\Gamma(s+\nu+1)}
\left[\mathcal{P}(s,t,\nu)-g(s^2+t^2+\nu s-st)\right]
\nonumber\\
&\times\left(
\frac{g}{x}+\frac{g^2}{x^2}\left((s+\nu)^2-t^2\right)+
+\frac{g^3}{x^3}\left(
\left((s+\nu)^2-t^2\right)^2 -2\left((s+\nu)^2+t^2\right)+1
\right)
+O\left(\frac{g^4}{x^4}\right)
\right).
\end{split}
\end{equation}
By comparing order by order in $g$ the following integral identities are obtained at order $O(g)$:
\begin{equation}\label{Identity1}
0=\oint\limits_{\Sigma}dt\oint\limits_{\Sigma}ds
\frac{\Gamma(-t)\Gamma(-s)x^{t+s+\nu}}{\Gamma(t+\nu+1)\Gamma(s+\nu+1)}
\frac1x\mathcal{P}(s,t,\nu)\ ,
\end{equation}
at order $O(g^2)$:
\begin{equation}\label{Identity2}
0=\oint\limits_{\Sigma}dt\oint\limits_{\Sigma}ds
\frac{\Gamma(-t)\Gamma(-s)x^{t+s+\nu}}{\Gamma(t+\nu+1)\Gamma(s+\nu+1)}
\left\{\frac{\mathcal{P}(s,t,\nu)}{x^2}\left((s+\nu)^2-t^2\right)
-\frac1x\left(s^2+t^2+\nu s-st\right)\right\},
\end{equation}
and at order $O(g^3)$:
\begin{equation}\label{Identity3}
\begin{split}
0=&\oint\limits_{\Sigma}dt\oint\limits_{\Sigma}ds
\frac{\Gamma(-t)\Gamma(-s)x^{t+s+\nu}}{\Gamma(t+\nu+1)\Gamma(s+\nu+1)}
\left\{
\frac{\mathcal{P}(s,t,\nu)}{x^3}\left(
\left((s+\nu)^2-t^2\right)^2 -2\left((s+\nu)^2+t^2\right)+1
\right)\right.
\\
&\quad\quad\quad\quad\quad\quad\quad\quad\quad
\quad\quad\quad\quad\quad\quad\quad\left.-\frac{1}{x^2}
\left(s^2+t^2+\nu s-st\right)\left((s+\nu)^2-t^2\right)
\right\}.
\end{split}
\end{equation}
We note in passing that these integral identities could be translated into identities among Bessel functions.

Reusing the expansion (\ref{WXX}) with $t-1$ instead of $t$ we obtain
from eq. (\ref{TransitionKuijlaarsZhang})
\begin{eqnarray}
\mathbb{S}(x,x;g)&=&
\frac{1}{g^2x(2\pi i)^2}\oint\limits_{\Sigma}dt\oint\limits_{\Sigma}ds
\frac{\Gamma(-t)\Gamma(-s)x^{t+s+\nu}}{\Gamma(t+\nu+1)\Gamma(s+\nu+1)}
\nonumber\\
&\times&\left[\mathcal{P}(s,t,\nu)-g(s^2+t^2+\nu s-st)\right]
\left(
\frac{g}{x}+\frac{g^2}{x^2}\left((s+\nu)^2-(t-1)^2\right)+O\left(g^3\right)
\right),
\end{eqnarray}
as $g\rightarrow 0$. After the repeated application of the identities (\ref{Identity1}) and (\ref{Identity2})
we find
\begin{equation}
\label{Slim}
\lim_{g\to0}{\mathbb{S}}(x,x;g)=
\frac{1}{(2\pi i)^2}\oint\limits_{\Sigma}dt\oint\limits_{\Sigma}ds
\frac{\Gamma(-t)\Gamma(-s)x^{t+s+\nu-3}}{\Gamma(t+\nu+1)\Gamma(s+\nu+1)}\ 2t\mathcal{P}(s,t,\nu).
\end{equation}
The first equality in Proposition \ref{PropositionDensity2} is thus proved. It remains to show that the contour integral in eq. (\ref{Slim}) above equals
the Bessel density $S^{\sqrt{\Bessel}}(x,x)$ eq. (\ref{SBesselXX}). For this purpose we need three further identities expressing Bessel functions in terms of contour integrals. These follow by simple differentiation of the identities of Proposition \ref{PropositionB} by $\xi\frac{\partial}{\partial \xi}$
\begin{equation}\label{Identity4}
\begin{split}
\frac{1}{2\pi i}\oint\limits_{\Sigma}\frac{dt\,t^2\Gamma(-t)}{\Gamma(\nu+1+t)}\xi^t=&
\ \xi^{-\frac{\nu}{2}+\frac{1}{2}}J_{\nu+1}\left(2\xi^{\frac{1}{2}}\right)
- \xi^{-\frac{\nu}{2}+1}J_{\nu+2}\left(2\xi^{\frac{1}{2}}\right)
,\\
\frac{1}{2\pi i}\oint\limits_{\Sigma}\frac{dt\,t^2(t+\nu)\Gamma(-t)}{\Gamma(\nu+1+t)}\xi^t=&
\ \xi^{-\frac{\nu}{2}+1}J_{\nu}\left(2\xi^{\frac{1}{2}}\right)
- \xi^{-\frac{\nu}{2}+\frac32}J_{\nu+1}\left(2\xi^{\frac{1}{2}}\right)
,\\
\frac{1}{2\pi i}\oint\limits_{\Sigma}\frac{dt\,t^2(t+\nu)(t+\nu-1)\Gamma(-t)}{\Gamma(\nu+1+t)}\xi^t=&
\ \xi^{-\frac{\nu}{2}+\frac{3}{2}}J_{\nu-1}\left(2\xi^{\frac{1}{2}}\right)
-\xi^{-\frac{\nu}{2}+2}J_{\nu}\left(2\xi^{\frac{1}{2}}\right).
\end{split}
\end{equation}
Inserting these results together with Proposition \ref{PropositionB} into eq. (\ref{Slim}), using eq. (\ref{MainPolynomial3}) we obtain
\begin{eqnarray}
\label{Sfinallim}
\lim_{g\to0}{\mathbb{S}}(x,x;g)&=&
\frac{-1}{2x^3}\left[
x^2 J_{\nu+1}\left(2x^{\frac{1}{2}}\right)J_{\nu-1}\left(2x^{\frac{1}{2}}\right)+x^{\frac32}
J_{\nu-1}\left(2x^{\frac{1}{2}}\right)J_{\nu}\left(2x^{\frac{1}{2}}\right)
-x^2J_{\nu}\left(2x^{\frac{1}{2}}\right)^2
\right.\nonumber\\
&& \quad\quad
-\nu x^{\frac32}
J_{\nu+1}\left(2x^{\frac{1}{2}}\right)J_{\nu}\left(2x^{\frac{1}{2}}\right)
-\nu x J_{\nu}\left(2x^{\frac{1}{2}}\right)^2
+\nu x^{\frac32}
J_{\nu+1}\left(2x^{\frac{1}{2}}\right)J_{\nu}\left(2x^{\frac{1}{2}}\right)
\nonumber\\
&&\left.\quad\quad
-x^2\left(J_{\nu+1}\left(2x^{\frac{1}{2}}\right)\right)^2
+x^2 J_{\nu+2}\left(2x^{\frac{1}{2}}\right)J_{\nu}\left(2x^{\frac{1}{2}}\right)
\right]\nonumber\\
&=&\frac{1}{x}\left[ J_{\nu}\left(2x^{\frac{1}{2}}\right)^2
-J_{\nu+1}\left(2x^{\frac{1}{2}}\right)J_{\nu-1}\left(2x^{\frac{1}{2}}\right)
\right] =\ S^{\sqrt{\Bessel}}(x,x)\ .
\end{eqnarray}
In the last step we have used several times the identity for Bessel functions
eq. (\ref{BesselId}). Proposition \ref{PropositionDensity2}
is thus proved.
\qed

\begin{appendix}
\section{Heuristic map to the Bessel-kernel
}\label{Heuristic}
In this appendix we present a heuristic map of the interpolating kernel to the Bessel-kernel as well as some additional checks.
It is heuristic as we first take the limit $\mu\to0$ and then $N\to\infty$.

Recall that the correlation kernel $K_N(x,y,\mu)$ can be expressed in terms of functions $P_n(x)$ and $Q_n(y)$, see
equation (\ref{K}). The functions $P_n(x)$ and $Q_n(y)$ are defined by equations (\ref{P1}) and (\ref{Q1}) respectively.
Let us obtain asymptotic formulae for $P_n(x)$ and $Q_n(y)$ as $\mu\rightarrow 0$.
\begin{prop}\label{PropositionPQAsymptoticsMu0} As $\mu\rightarrow 0$, the following asymptotic formulae hold true
for $x,y>0$:
$$
P_n(x)\sim\frac{(-1)^n\Gamma^2(n+1)}{(2\pi)^{\frac{1}{2}}}L_n^{(\nu)}\left(2x^{\frac{1}{2}}\right)
\frac{\e^{-x^{\frac{1}{2}}}}{x^{\frac{1}{4}}}\e^{\frac{x^{\frac{1}{2}}}{\mu}},
$$
and
$$
Q_n(y)\sim\frac{(-1)^n(2\pi)^{\frac{1}{2}}\left(2y^{\frac{1}{2}}\right)^{\nu}}{\Gamma(n+1)\Gamma(n+1+\nu)}L_n^{(\nu)}\left(2y^{\frac{1}{2}}\right)
\frac{\e^{-y^{\frac{1}{2}}}}{y^{\frac{1}{4}}}\e^{-\frac{y^{\frac{1}{2}}}{\mu}}.
$$
\end{prop}
\begin{proof}
We use the contour integral representations for $P_n(x)$ and $Q_n(y)$ obtained in Proposition \ref{PropositionPQContourIntegralRepresentation},
and replace the Bessel functions in the formulas for $P_n(x)$ and $Q_n(y)$ by their asymptotic expressions for  large values of the arguments derived in Proposition \ref{PropositionBesselLargeArgumentAsymtotics}, eq. (\ref{IKasymptotic}).
Next, a computation using the Residue Theorem gives
$$
\frac{1}{2\pi i}\oint\limits_{\Sigma_n}dt
\frac{\Gamma(t-n)\left(\frac{2}{1-\mu}x^{\frac{1}{2}}\right)^t}{\Gamma(\nu+1+t)\Gamma(t+1)}=\frac{(-1)^n}{\Gamma(\nu+1+n)}L_n^{(\nu)}\left(\frac{2}{1-\mu}x^{\frac{1}{2}}\right),
$$
and
$$
\frac{1}{2\pi i}
\oint\limits_{\Sigma_n}ds
\frac{\Gamma(s-n)\left(\frac{2}{1+\mu}y^{\frac{1}{2}}\right)^{s+\nu}}{\Gamma(\nu+1+s)\Gamma(s+1)}=
\frac{(-1)^n}{\Gamma(\nu+1+n)}\left(\frac{2}{1+\mu}y^{\frac{1}{2}}\right)^{\nu}L_n^{(\nu)}\left(\frac{2}{1+\mu}y^{\frac{1}{2}}\right),
$$
with
$$
L_n^{(\nu)}(x)=\sum_{k=0}^n\frac{(-1)^k\Gamma(n+\nu+1)}{(n-k)!\Gamma(\nu+k+1)}
\frac{x^k}{k!}\ .
$$
In this way we obtain the asymptotic expressions for $P_n(x)$ and $Q_n(y)$  in the statement of this proposition. Alternatively we could have directly inserted the asymptotic expressions from Proposition \ref{PropositionBesselLargeArgumentAsymtotics}, eq. (\ref{IKasymptotic})
into the definitions eqs. (\ref{P1}) and (\ref{Q1}).
\end{proof}
\begin{rem} To check the formulae in Proposition \ref{PropositionPQAsymptoticsMu0} we  compute $\int_0^{\infty}P_n(x)Q_n(x)dx$. Using the orthogonality relation for the Laguerre polynomials, we see that this integral is indeed equal to $1$ as it should be.  Also, Proposition \ref{PropositionPQAsymptoticsMu0} enables us to check
the recurrence relations for $P_n(x)$ and $Q_n(y)$
asymptotically at small $\mu$
obtained by the authors in \cite[Prop. 3.5]{AkemannStrahov}.
\end{rem}
Next, we obtain the asymptotics of the kernel $K_N(x,y;\mu)$ as $\mu\rightarrow0$ in terms of the Laguerre polynomials.
\begin{prop}\label{PropositionKNMU} As $\mu\rightarrow 0$,
\begin{equation}
\begin{split}
K_N(x,y;\mu)\sim&\sum\limits_{n=0}^{N-1}\frac{\Gamma(n+1)}{\Gamma(n+1+\nu)}L_n^{(\nu)}\left(2x^{\frac{1}{2}}\right)L_n^{(\nu)}\left(2y^{\frac{1}{2}}\right)\\
&\times\left[\left(2x^{\frac{1}{2}}\right)^{\frac{\nu}{2}}\frac{\e^{-x^{\frac{1}{2}}}}{x^{\frac{1}{4}}}\right]
\left[\left(2y^{\frac{1}{2}}\right)^{\frac{\nu}{2}}\frac{\e^{-y^{\frac{1}{2}}}}{y^{\frac{1}{4}}}\right]
\left(\frac{y^{\frac{1}{2}}}{x^{\frac{1}{2}}}\right)^{\frac{\nu}{2}}\frac{\e^{\frac{x^{\frac{1}{2}}}{\mu}}}{\e^{\frac{y^{\frac{1}{2}}}{\mu}}}\\
=& K_N^{\Laguerre}(x,y).
\end{split}
\end{equation}
\end{prop}
\begin{proof} Once the asymptotics of the functions $P_n(x)$ and $Q_n(y)$ is given by Proposition \ref{PropositionPQAsymptoticsMu0}, the asymptotic formula for the correlation kernel
$K_N(x,y;\mu)$ in the statement of  Proposition \ref{PropositionKNMU} can be derived  straightforwardly. By direct comparison to eq. (\ref{KLaguerre})
we see that this agrees with the kernel of Laguerre polynomials defined there.
\end{proof}

Because of the agreement of our kernel in the limit $\mu\to0$ with the kernel of Laguerre polynomials just stated it is clear that it satisfies the same Christoffel-Darboux formula (\ref{Kmu0}), as well as the same Bessel asymptotics eq. (\ref{BesselKernel}). Thus we arrive at the same statement as  Theorem
\ref{TheoremTransitionToBesselKernel}, (a).

We note however, that the order of limits $\mu\to0$ and $N\to\infty$ is subtle as in the double scaling limit with $\mu=g/N$ and $N\to\infty$ with $g>0$ we get a different answer, namely  Theorem
\ref{TheoremTransitionToBesselKernel}, (b).

\begin{rem}
We can obtain the same formula for $K_N(x,y;\mu)$
as the right-hand side of
eq. (\ref{Kmu0}) (and thus the Bessel-kernel)
 starting from Proposition \ref{TheoremChristoffelDarboux}, using
explicit formula for the coefficients $a_{-2,N}$, $a_{-1,N}$, $a_{1,N}$ and $a_{2,N}$ in the statement of Proposition \ref{TheoremChristoffelDarboux},
and the asymptotic formulae for $P_n(x)$ and $Q_n(y)$ of Proposition \ref{PropositionPQAsymptoticsMu0}.
Namely, starting from Proposition \ref{TheoremChristoffelDarboux} we find
\begin{equation}\label{Kexpanded}
\begin{split}
K_N(x,y;\mu)\sim&\frac{\Gamma(N+1)}{\Gamma(N+\nu)}\frac{1}{4(x-y)}\\
&\biggl\{-(N+\nu-1)L_{N-2}^{(\nu)}\left(2x^{\frac{1}{2}}\right)L_{N}^{(\nu)}\left(2y^{\frac{1}{2}}\right)+(N+\nu-1)L_{N}^{(\nu)}\left(2x^{\frac{1}{2}}\right)
L_{N-2}^{(\nu)}\left(2y^{\frac{1}{2}}\right)\\
&-(N+1)L_{N-1}^{(\nu)}\left(2x^{\frac{1}{2}}\right)L_{N+1}^{(\nu)}\left(2y^{\frac{1}{2}}\right)+(N+1)L_{N+1}^{(\nu)}\left(2x^{\frac{1}{2}}\right)
L_{N-1}^{(\nu)}\left(2y^{\frac{1}{2}}\right)\\
&+2(2N+\nu)L_{N}^{(\nu)}\left(2x^{\frac{1}{2}}\right)L_{N-1}^{(\nu)}\left(2y^{\frac{1}{2}}\right)-2(2N+\nu)L_{N-1}^{(\nu)}\left(2x^{\frac{1}{2}}\right)
L_{N}^{(\nu)}\left(2y^{\frac{1}{2}}\right)\biggr\}\\
&\times\left[\left(2x^{\frac{1}{2}}\right)^{\frac{\nu}{2}}\frac{\e^{-x^{\frac{1}{2}}}}{x^{\frac{1}{4}}}\right]
\left[\left(2y^{\frac{1}{2}}\right)^{\frac{\nu}{2}}\frac{\e^{-y^{\frac{1}{2}}}}{y^{\frac{1}{4}}}\right]
\left(\frac{y^{\frac{1}{2}}}{x^{\frac{1}{2}}}\right)^{\frac{\nu}{2}}\frac{\e^{\frac{x^{\frac{1}{2}}}{\mu}}}{\e^{\frac{y^{\frac{1}{2}}}{\mu}}},
\end{split}
\end{equation}
as $\mu\rightarrow 0$. Using the recurrence relation for the Laguerre polynomials
$$
xL_n^{(\nu)}(x)=-(n+1)L_{n+1}^{(\nu)}(x)+(2n+\nu+1)L_n^{(\nu)}(x)-(N+\nu)L_{n-1}^{(\nu)}(x),
$$
it is not hard to see that equation (\ref{Kexpanded}) is equivalent to the right-hand side of equation (\ref{Kmu0}).
In particular, such a calculation confirms that the formulae
stated in Proposition \ref{TheoremChristoffelDarboux} hold true.
\end{rem}

\section{Justification for interchanging sum and integral in the proof of Proposition \ref{Prop4.1}
}\label{AppB}

Let $A>0$, $B>0$, $\nu>0$, and $c>0$ be fixed real positive numbers.
Set
$$
f_n(z)=\frac{(-1)^nB^n\Gamma(z+\nu+n)A^z\Gamma(z)}{n!\Gamma(1+\nu+n)},
$$
where $z=c+iy$, $y\in\left(-\infty,+\infty\right)$, and where $n=0,1,\ldots$
To justify the interchange of the sum and the integral in the proof of Proposition 4.1
means
to prove that
\begin{equation}\label{MainEquation}
\int\limits_{c-i\infty}^{c+i\infty}\left(\sum\limits_{n=0}^{\infty}f_n(z)\right)dz=
\sum\limits_{n=0}^{\infty}\left(\int\limits_{c-i\infty}^{c+i\infty}f_n(z)dz\right).
\end{equation}
Recall that the gamma function $\Gamma(x)$ is positive for positive $x>0$. In what follows  we will use the following inequalities  for the gamma function:\\
(I) For $x>0$ we have (see Ref. \cite{NIST}, Section 5.6, inequality 5.6.6):
\begin{equation}\label{G1}
\left\vert\Gamma(x+iy)\right\vert\leq \Gamma(x).
\end{equation}
(II) If $z=x+iy$, and $x\geq 0$, then (see Ref. \cite{NIST}, Section 5.6, inequality 5.6.9):
\begin{equation}\label{G2}
\left\vert\Gamma(z)\right\vert\leq\left(2\pi\right)^{\frac{1}{2}}|z|^{x-\frac{1}{2}}e^{-\frac{\pi|y|}{2}}e^{\frac{1}{6|z|}}.
\end{equation}
Choose $\Lambda>0$. First, let us prove that
\begin{equation}\label{MainEquationRestricted}
\int\limits_{c-i\Lambda}^{c+i\Lambda}\left(\sum\limits_{n=0}^{\infty}f_n(z)\right)dz=
\sum\limits_{n=0}^{\infty}\left(\int\limits_{c-i\Lambda}^{c+i\Lambda}f_n(z)dz\right).
\end{equation}
Since the integration contour in the integrals above is finite, it is enough to show that
$$
\sum\limits_{n=0}^{\infty}f_n(z)
$$
converges uniformly on $\left\{z:z=c+iy,\; -\Lambda\leq y\leq \Lambda\right\}$. But we have
$$
\vert f_n(z)\vert=A^{c}\frac{B^n}{n!}\frac{\vert\Gamma(z+\nu+n)\vert}{\Gamma(1+\nu+n)}|\Gamma(z)|
\leq c_1\frac{B^n}{n!}\frac{\Gamma(c+\nu+n)}{\Gamma(1+\nu+n)},
$$
where the constant $c_1$ does not depend on $y$, and where we have used inequality (\ref{G1}).
Since the series
$$
\sum\limits_{n=0}^{\infty}\frac{B^n}{n!}\frac{\Gamma(c+\nu+n)}{\Gamma(1+\nu+n)}
$$
converges, we conclude that $\sum_{n=0}^{\infty}f_n(z)$ converges uniformly on $\left\{z:z=c+iy,\; -\Lambda\leq y\leq \Lambda\right\}$.
Thus equation (\ref{MainEquationRestricted}) holds true.

Now we can write
\begin{equation}\label{ManyParts}
\begin{split}
\int\limits_{c-i\Lambda}^{c+i\Lambda}\left(\sum\limits_{n=0}^{\infty}f_n(z)\right)dz&=\sum\limits_{n=0}^{\infty}\left(\int\limits_{c-i\Lambda}^{c+i\Lambda}f_n(z)dz\right)\\
&=\sum\limits_{n=0}^{\infty}\left(\int\limits_{c-i\infty}^{c+i\infty}f_n(z)dz\right)-\sum\limits_{n=0}^{\infty}\left(\int\limits_{c+i\Lambda}^{c+i\infty}f_n(z)dz\right)
-\sum\limits_{n=0}^{\infty}\left(\int\limits_{c-i\infty}^{c-i\Lambda}f_n(z)dz\right).
\end{split}
\end{equation}
Set
$$
J_1(\Lambda)=\sum\limits_{n=0}^{\infty}J_1^{(n)}(\Lambda),\;\;\;J_1^{(n)}(\Lambda)=\int\limits_{c+i\Lambda}^{c+i\infty}f_n(z)dz,
$$
and
$$
J_2(\Lambda)=\sum\limits_{n=0}^{\infty}J_2^{(n)}(\Lambda),\;\;\;J_2^{(n)}(\Lambda)=\int\limits_{c-i\infty}^{c-i\Lambda}f_n(z)dz.
$$
The principal observation is that in order to see that equation (\ref{MainEquation}) holds true it is enough to check that
$$
\underset{\Lambda\rightarrow\infty}{\lim}J_1(\Lambda)=0,\;\;\;\underset{\Lambda\rightarrow\infty}{\lim}J_2(\Lambda)=0,
$$
as it follows from equation (\ref{ManyParts}). Let us estimate $J_1(\lambda)$. We have
$$
|J_1(\Lambda)|\leq\sum\limits_{n=0}^{\infty}|J_1^{(n)}(\Lambda)|,
$$
and
$$
|J_1^{(n)}(\Lambda)|\leq \frac{B^nA^c}{n!\Gamma(1+\nu+n)}\int\limits_{\Lambda}^{+\infty}|\Gamma(c+iy+\nu+n)|
|\Gamma(c+iy)|dy.
$$
Moreover,
$$
|\Gamma(c+\nu+n+iy)|\leq\Gamma(c+\nu+n),
$$
as it follows from inequality (\ref{G1}). Using inequality (\ref{G2}) we obtain
$$
|\Gamma(c+iy)|\leq(2\pi)^{\frac{1}{2}}\left(\sqrt{c^2+y^2}\right)^{c-\frac{1}{2}}e^{-\frac{\pi y}{2}}e^{\frac{1}{6\sqrt{c^2+y^2}}}
\leq c_2e^{-c_3y},
$$
for some positive constants $c_2$, $c_3$, and for $y>0$. This gives
$$
|J_1^{(n)}(\Lambda)|\leq c_4\frac{B^n\Gamma(c+\nu+n)}{n!\Gamma(1+\nu+n)}e^{-c_3\Lambda}.
$$
Since the series
$$
\sum\limits_{n=0}^{\infty}\frac{B^n}{n!}\frac{\Gamma(c+\nu+n_)}{\Gamma(1+\nu+n)}
$$
converges, we conclude that $\underset{\Lambda\rightarrow\infty}{\lim}J_1(\Lambda)=0$ holds true indeed.
In the same way we can show that $\underset{\Lambda\rightarrow\infty}{\lim}J_2(\Lambda)=0$ holds too.

\end{appendix}
%

\end{document}